%% file: Modal_doctrines.tex
\documentclass[runningheads]{llncs}
\usepackage{times}
\usepackage{array}
\usepackage{amsmath}
\usepackage{amsfonts}
\usepackage{float}
\usepackage{fancyvrb}
\usepackage{amssymb}
\usepackage{bussproofs}
\usepackage[all,2cell]{xy}
\UseAllTwocells
\xyoption{v2}
\NoResizing
\UseResizing
\DeclareMathAlphabet{\mathpzc}{OT1}{pzc}{m}{it}
\usepackage{mathrsfs}

\usepackage{tikz}
\usepackage{tikzit}

\usetikzlibrary{shapes}

\usepackage{freetikz}
\usetikzlibrary{decorations.markings,positioning,patterns}
\usetikzlibrary{shadows}
\usepackage{array}

\usetikzlibrary{automata, positioning, arrows}

\tikzstyle{new style 0}=[fill={rgb,255: red,191; green,63; blue,66}, draw=black, shape=circle]
\tikzstyle{new style 1}=[fill=white, draw=black, shape=rectangle]
\tikzstyle{freccia_tratteggiata}=[->, dashed]
\tikzstyle{new edge style 0}=[->, dashed, draw=blue]
\tikzstyle{new edge style 1}=[->, dashed, draw={rgb,255: red,0; green,128; blue,128}]
\tikzstyle{freccia}=[->]

\makeatletter
\providecommand*{\Dashv}{%
  \mathrel{%
    \mathpalette\@Dashv\vDash
  }%
}
\newcommand*{\@Dashv}[2]{%
  \reflectbox{$\m@th#1#2$}%
}
\makeatother
\include{macro_davide}
\begin{document}

\title{A Presheaf Semantics\\ for Quantified Temporal Logics\thanks{
Research partially supported by the Italian MIUR project PRIN 2017FTXR7S “IT-MaTTerS”.}}

\author{Fabio Gadducci\inst{}\orcidID{0000-0003-0690-3051} \and
Davide Trotta\inst{}\orcidID{0000-0003-4509-594X } }

\institute{University of Pisa, Department of Computer Science, Pisa, Italy\\
\email{fabio.gadducci@unipi.it, trottadavide92@gmail.com}}

\authorrunning{F. Gadducci and D.~Trotta}

\maketitle

\begin{abstract}
Temporal logics stands for a widely adopted family of formalisms for the verification of computational devices, enriching propositional logics by operators predicating on the step-wise behaviour of a system. 
Its quantified extensions allow
to reason on the properties of the individual components of the system at hand. The expressiveness of the resulting logics poses problems in correctly identifying a semantics that exploits its features without resorting to the imposition of restrictions on the acceptable behaviours.
In this paper we address this issue by means of counterpart models 
and relational presheaves. 

\end{abstract}

\section{Introduction}

The words ``temporal logics'' stand for a widely adopted family of formalisms for the specification and verification of computational devices, ranging from stand-alone programs to large-scale  systems, which find applications in diverse areas such as synthesis, planning and knowledge representation, see e.g.~\cite{Reif85,Cardelli2002} among many others. 
Usually, these logics have a propositional fragment at the core, which is extended by operators predicating on the step-wise behaviour of a system. 
The framework proved extremely effective, and after the foundational work carried out since Pnueli’s seminal paper~\cite{Pnueli77}, 
research focused on the 
development of techniques for the verification of properties specified via such logics, see e.g.~\cite{BaldanCorradini2006,Meseguer2008} for their emphasis on (graph) rewriting.

The standard model for temporal logics is represented by transition systems (also known as Kripke frames in logical jargon), where the states represent possible configurations of a device and the transition relation their possible evolution. Most often, the states have some (algebraic) structure and one is interested not just on the topology of each reachable state, but on the fate of its individual items as well. To this end, 
the use of quantified temporal logics have been advocated. A typical example are graph logics, where the states are graphs (i.e. algebras of a graph signature) and the transition relation is given by a family of (partial) graph morphisms~\cite{COURCELLE199012,COURCELLE97,DAWAR2007263}.
Unfortunately, such logics are in general not decidable~\cite{Franconi2003FixpointEO,Hodkinson2001} and, as a consequence, many efforts have been devoted to the definition of logics (or the identification
of fragments) that sacrifice expressiveness in favour of computability and efficiency.

Besides these practical considerations, notice that also the characteristics of the semantical models for such logics are not clearly cut. Consider for instance a model with two states $s_0$, $s_1$, a transition from $s_0$ to $s_1$ and another transition going backward, and an item $i$ that appears in $s_0$ only. Is item $i$ being destroyed and (re)created again and again? Or is it just an identifier
that is being reused? The issue is denoted in the literature as the \emph{trans-world identity problem} (see \cite{Hazen1979} as well as \cite{Belardinelli2006QuantifiedML} for a survey of the related philosophical issues). An often adopted solution is to choose a set of universal items, which are used to form each state. It is then obvious how it is possible to refer to the
same element across states. 
However, these solutions are not perfectly suited to model systems with dynamic allocation and
deallocation of components. Consider again the above example. The problem is that item $i$ belongs to the
universal domain, and hence it is exactly the same after every deallocation in state $s_1$. But intuitively, every instance of $i$ should instead be considered to be distinct (even if syntactically equivalent). 

The solution advanced by Lewis~\cite{CounterpartTheoryLewis} is based on what is called the counterpart paradigm: instead of a universal set of items, states are connected by (possibly partial) morphisms, so that state items are local and state evolution may account for their deallocation and
merging.
%
However, as far as we are aware, the use of counterpart semantics for quantified temporal or modal logics has been sparse, and even more so the attempts for their categorical presentation.
The paper builds on the set-theoretical description of a counterpart semantics for a modal logic with second-order quantifiers introduced in~\cite{CounterpartSemanticsGLV}, with the goal of explaining how that model admits a natural generalization and presentation in a categorical setting, and how it can be adapted to
offer a counterpart semantics for (linear time) temporal logics with second-order quantifiers.

From a technical perspective, the starting point for our semantics was the hyperdoctrine presentation of first-order logics, as originally described by Lawvere in his work on categories with equational structure~\cite{Lawvere1969,Lawvere1969b}. More precisely, the direct inspiration was the presheaf model for modal logics with first-order quantifiers presented in~\cite{GhilardiMeloni1988}. Our work extends and generalizes the latter proposal in a few directions.
First of all, the focus on temporal logics, with an explicit operator for next step, required to tweak the original proposal by Ghilardi and Meloni by equipping our models with a chosen family of arrows, which represent the basics steps of a system. Furthermore, the choice of a counterpart semantics forced the transition relation between worlds to be given by families of partial morphisms between the algebras forming each world: this was modelled by using relational presheaves, instead of functional ones, as the main tool. Related to this, tackling second-order quantification required additional effort since relational presheaves do not form a topos in general \cite{Niefield2010}.



The paper has the following structure. Section~2 recalls some basic properties of multi-sorted algebras, which describe the  structure of our worlds. Section~3 presents relational presheaves, showing how the latter capture a generalised notion of transition system, where a transition step is given by a partial homomorphism, and how they support suitable second-order operators. Section~4 introduces our logics, a monadic second-order extension of classical linear time temporal logics, and Section~6 finally shows how to provide it with a counterpart semantics, thanks to the categorical set-up of the previous sections. This paper is rounded up by a concluding section and by a running example, highlighting the features of the chosen logics.
\paragraph{Related works.} 
Functional presheaves can be seen as the categorical abstraction of Krikpke frames. This idea was introduced in \cite{GhilardiMeloni1988,GhilardiMeloni1990}. 
Developing this intuition, relational presheaves (see \cite{Niefield2010} for an analysis of the structure of the associated category) can be thought of as a categorification of \emph{counterpart} Kripke frames in the sense of Lewis \cite{CounterpartTheoryLewis}. An in-depth presentation of classical counterpart semantics is in \cite{Hazen1979,Belardinelli2006QuantifiedML}.

In this work we use relational preaheaves to provide a categorical account to the notion of counterpart model defined in \cite{CounterpartSemanticsGLV}, introducing what we call counterpart $\mW$-models. We then specialize our models to deal with temporal logics, equipping them with what we call temporal structures.

As we already mentioned in the introduction, many authors considered quantified temporal logics and have addressed their decidability and complexity issues.
Since most of our examples are motivated by applications to graph rewriting, we just refer to \cite{BaldanCorradini2006,Rensink06}, among others.
Much less explored, though, is the side of categorical semantics.
More specifically, we are aware of a topos-theoretical description of a semantics for modal logics in~\cite{Awodey04} and a presentation in terms of Lawvere's doctrines in~\cite{Braner2007FirstorderML}. Moreover, for connections between the areas of coalgebra and modal logic we refer to \cite{jacobs2002,Jacobs2001}.
Notice that the approaches presented in these works generalise in the categorical setting the usual Kripke-style semantics, but not the counterpart one.



\section{Some notions on multi-sorted algebras}
We begin by recalling the definition of many-sorted algebras and their homomorphisms, which lies at the basis of the structure of our worlds.

\begin{definition}
A \tbf{many-sorted signature} $\Sigma$ is a pair $(S_{\Sigma},F_{\Sigma})$ given by a set of \tbf{sorts} $S_{\Sigma}=\{\tau_1,\dots,\tau_m\}$ and by a set $F_{\Sigma}=\{\freccia{\tau_1\times \cdots \times \tau_m}{f_{\Sigma}}{\tau}\; |\; \tau_i, \tau \in S_{\Sigma}\}$ of \tbf{function symbols} typed over $S_{\Sigma}^*$.
\end{definition}

\begin{definition}
A \tbf{many-sorted algebra} $\MSalgebra{A}$ with signature $\Sigma$, i.e. a $\Sigma$-algebra, is a pair $(A,F_{\Sigma}^{\MSalgebra{A}})$ such that
\begin{itemize}
    \item $A$ is a set whose elements are typed over $S_{\Sigma}$;
    \item $F_{\Sigma}^{\MSalgebra{A}}$ is a family of typed functions $\{\freccia{A_{\tau_1}\times \cdots \times A_{\tau_m}}{f_{\Sigma}^{\MSalgebra{A}}}{A_{\tau}}\; |\; f_{\Sigma} \in F_{\Sigma} \wedge \freccia{\tau_1\times\cdots \times \tau_m}{f_{\Sigma}}{\tau}\}$.
\end{itemize}
\end{definition}

Notice that we denoted by $A_{\tau}$ the set $\{a\in A\;|\; a:\tau\}$ of elements of $A$ with type $\tau$.

\begin{definition}
Given two $\Sigma$-algebras $\MSalgebra{A}$ and $\MSalgebra{B}$, a \tbf{(partial) homomorphism}
$\rho$ is a family of (partial) functions $\rho:=\{\frecciaparziale{A_{\tau}}{\rho_{\tau}}{B_{\tau}}\;|\; \tau\in S_{\Sigma}\}$ typed over $S_{\Sigma}$ such that for every function symbol $\freccia{\tau_1\times\cdots\times \tau_m}{f_{\Sigma}}{\tau}$ and for every list of elements $(a_1,\dots,a_m)$, if $\rho_{\tau_i}$ is defined for the element $a_i$ of type $\tau_i$, then $\rho_{\tau}$ is defined for the element $f_{\Sigma}^{\MSalgebra{A}}(a_1,\dots,a_m)$ and $\rho_{\tau}(f_{\Sigma}^{\MSalgebra{A}}(a_1,\dots,a_m))=f_{\Sigma}^{\MSalgebra{B}}(\rho_{\tau_1}(a_1),\dots,\rho_{\tau_m}(a_m))$.

\end{definition}

\begin{example}[Graph Algebra]\label{esempio sigma algebra grafi}
Let us consider the signature $\Sigma_{Gr}=(S_{Gr},F_{Gr})$ for directed graphs. The set $S_{Gr}$ consists of the sorts of nodes $\tau_N$ and edges $\tau_E$, while the set $F_{Gr}$ is composed by the function symbols
$\freccia{\tau_E}{s,t}{\tau_N}$, which determine, respectively, the source and the target node of an edge. In this case, a $\Sigma_{Gr}$-Algebra $\MSalgebra{G}$ is a directed graph and a homomorphism of $\Sigma_{Gr}$-algebras is exactly a (partial) morphism of directed graphs.
In Figure \ref{Figura grafi } we find the visual representations for three graphs $\mathbf{G}_0$, $\mathbf{G}_1$ and $\mathbf{G}_2$.

\begin{figure}[H]\label{Figura grafi }
\centering

\begin{tikzpicture}[scale=0.40]
	\begin{pgfonlayer}{nodelayer}
		\node [style=new style 0] (0) at (-11, 6) {$n_0$};
		\node [style=new style 0] (1) at (-5, 6) {$n_1$};
		\node [style=new style 0] (2) at (-8, 2) {$n_2$};
		\node [style=new style 1] (3) at (-11, 4) {$e_2$};
		\node [style=new style 1] (4) at (-5, 4) {$e_1$};
		\node [style=new style 1] (5) at (-8, 8) {$e_0$};
		\node [style=new style 0] (6) at (-1, 5) {$n_3$};
		\node [style=new style 0] (7) at (4.75, 5) {$n_4$};
		\node [style=new style 1] (8) at (2, 8) {$e_3$};
		\node [style=new style 1] (9) at (2, 2) {$e_4$};
		\node [style=new style 0] (10) at (11, 6.25) {$n_5$};
		\node [style=new style 1] (13) at (11, 3.5) {$e_5$};
	\end{pgfonlayer}
	\begin{pgfonlayer}{edgelayer}
		\draw (2) to (3);
		\draw (0) to (5);
		\draw (1) to (4);
		\draw [style=freccia] (4) to (2);
		\draw [style=freccia] (3) to (0);
		\draw [style=freccia] (5) to (1);
		\draw (6) to (8);
		\draw (7) to (9);
		\draw [style=freccia] (9) to (6);
		\draw [style=freccia] (8) to (7);
		\draw [bend right=60, looseness=1.50] (10) to (13);
		\draw [style=freccia, bend right=60, looseness=1.50] (13) to (10);
	\end{pgfonlayer}
\end{tikzpicture}

\caption{Three graphs: $\mathbf{G}_0$ (left), $\mathbf{G}_1 $ (middle) and $\mathbf{G}_2$ (right)}
\end{figure}
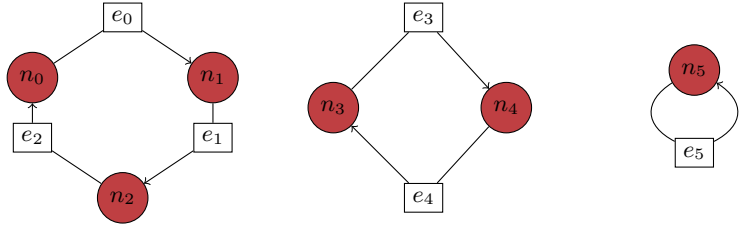
\end{example}
Let $\Sigma$ be a many-sorted signature, and let us fix disjoint, countably infinite sets $X_{\tau}$ of variables for each sort symbol of $S_{\Sigma}$. We let $x:\tau$ indicate an element $x$ of the set $X_{\tau}$. In order to introduce the notion of term, we take into account signatures $\Sigma_X$ obtained by extending a many-sorted signature $\Sigma$ with a denumerable set $X$ of variables typed over $S_{\Sigma}$. This allows to introduce the notion of \emph{open terms} over $\Sigma_X$ in the usual way. 
\begin{definition}
Let $\Sigma$ be a many-sorted signature and $X$ a denumerable set of individual variables typed over $S_{\Sigma}$. The many-sorted set $T(\Sigma_X)$ of \tbf{terms} obtained from $\Sigma_X$ is the smallest such that 
\begin{prooftree}
\AxiomC{}
\UnaryInfC{$X\subseteq T(\Sigma_X)$}
\end{prooftree}
\begin{prooftree}
\AxiomC{$ \freccia{\tau_1\times \cdots \times \tau_m}{f_{\Sigma}}{\tau}\in F_{\Sigma}$}
\AxiomC{$t_i:\tau_i\in T(\Sigma_X)$ for $i=1,\dots,m$}
\BinaryInfC{$f_{\Sigma}(t_1,\dots,t_m):\tau \in T(\Sigma_X)$}
\end{prooftree}

\end{definition}

Finally, we recall the notion of \emph{context} and \emph{term-in-context} over a signature $\Sigma_X$. 
\begin{definition}
A \tbf{context} $\Gamma$ over a signature $\Sigma_X$ is a finite list $\context{x}{\tau}{n}$ of (variable, sort)-pairs, subject to the condition that $x_1,\dots,x_n$ are \emph{distinct}.
\end{definition}
A \emph{term-in-context} takes the form $ \termincontext{t}{\tau}{\Gamma}$
where $t$ is a term, $\tau$ is a sort, and $\Gamma$ is a context over the given signature $\Sigma_X$. The well-formed terms-in-context are inductively generated by the two rules

\begin{prooftree}
\AxiomC{}
\UnaryInfC{$\termincontext{x}{\tau}{\Gamma',x:\tau, \Gamma}$}
\end{prooftree}
\begin{prooftree} 
\AxiomC{$\termincontext{t_1}{\tau_1}{\Gamma}\cdots  \termincontext{t_m}{\tau_m}{\Gamma}$}
\UnaryInfC{$\termincontext{f_{\Sigma}(t_1,\dots,t_m)}{\tau}{\Gamma}$}
\end{prooftree}

\noindent
where $\freccia{\tau_1\times\cdots\times \tau_m}{f_{\Sigma}}{\tau}$ is a function symbol of $F_{\Sigma}$.

\begin{definition}
Given a many-sorted signature $\Sigma$, its \tbf{syntactic category} or category of contexts $\SynCat{\Sigma}$ is defined as follows
\begin{itemize}
    \item its objects are $\alpha$-equivalence class of contexts;
\item a morphism $\freccia{\Gamma}{\gamma}{\Gamma'}$, where $\Gamma'=\context{y}{\tau}{m}$, is specified by an equivalence class of lists of the form $\gamma=[t_1,\dots,t_m]$ of terms over $\Sigma_X$ such that $\termincontext{t_i}{\tau_i}{\Gamma}$ holds for $i=1,\dots,m$.
\end{itemize} 
The composition of two morphisms is formed by making substitutions.
\end{definition}

Recall that two contexts are $\alpha$-equivalent if they differ only in their variables; in other words, the list of sorts occurring in each context
are equal (and in particular, the contexts are of equal length).
Requiring the objects of the syntactic category to be $\alpha$-equivalence classes of contexts ensures that the category $\SynCat{\Sigma}$ has finite products, given by the concatenations of contexts. As for morphisms, the equivalence relation we consider is the one identifying two terms if they are equal up to substitution.

\section{A categorical perspective of counterpart models} 
Kripke  semantics  is  widely  used  to  assign  a  meaning  to  modal  languages \cite{HandbookModalLogic2007};  it  stems  from  Leibniz’s  intuition  of  defining  necessity  as  truth  in  every  possible  world. We start this section recalling the notion of \emph{Kripke frame}, widely used
in the analysis of modality, and its development into counterpart semantics.
\begin{definition}
A \tbf{Kripke frame} is a 4-tuple $\angbr{W,R}{D,d}$ defined as
\begin{itemize}
    \item $W$ is a non-empty set;
    \item $R$ is a binary relation of $W$;
    \item $D$ is a function assigning to every $w\in W$ a non-empty set $D(w)$ such that if $wRw'$ then $D(w)\subseteq D(w')$;
    \item $d$ is a function assigning to every $w\in W$ a set $d(w)\subseteq D(w)$.
\end{itemize}
\end{definition}
The set $W$ is intuitively interpreted as the domain of \emph{possible worlds}, whereas $R$ is the
\emph{accessibility relation} among worlds. Each \emph{outer domain} $D(w)$ contains the objects of which
it makes sense to talk about in $w$, while each inner domain
$d(w)$ contains the \emph{individuals} actually existing in $w$.

Notice that Kripke frames are usually assumed to satisfy the \emph{increasing outer domain condition}, i.e. for all $w,w'\in W$, if $wRw'$ then $D(w)\subseteq D(w')$, but this condition could represent a strong, and not enough  motivated, constraint both philosophically and from the point of view of applications. Over the years, this condition has been the heart of several discussions and controversies.



In particular, Lewis denied the possibility of identifying the same individual across possible worlds, and he substituted the notion of trans-world identity with a \emph{counterpart relation} $C$, introducing counterpart theory \cite{CounterpartTheoryLewis}.
In order to assign a meaning to terms \emph{necessary} and \emph{possible} according to Lewis’ counterpart theory, Kripke frames are enriched by a function $C$ such that for all $w,w\in W$, $C_{w,w'}$ is a relation on $D(w)\times D(w')$, intuitively interpreted as the counterpart relation.
\begin{definition}
A \tbf{counterpart frame} is a 5-tuple $\angbr{W,R}{D,d,C}$ defined as
\begin{itemize}
    \item  $W,R,D,d$ are defined as for Kripke frames, but $D$ need not to satisfy the increasing outer domain condition;
    \item $C$ is a function assigning to every 2-tuple $\angbr{w}{w'}$ a subset of $D(w)\times D(w')$.
\end{itemize}
\end{definition}
As anticipated, Kripke-like solutions are not perfectly suited to model, for example, systems with \emph{dynamic allocation and
deallocation} of components. This will be clearer later, when we present Example \ref{esempio deallocation}.

Alternative solutions based on counterpart relations are e.g. introduced in \cite{CounterpartSemanticsGLV}, where the authors introduce a novel approach to the semantics of quantified $\mu$-calculi, considering a sort of
labeled transition systems as semantic domain (called counterpart models), where states are algebras
and transitions are defined by counterpart relations (a family of partial homomorphisms) between
states.

We conclude this section recalling the notion of counterpart model from \cite{CounterpartSemanticsGLV}.
\begin{definition}
\label{def counterpart model}
Let $\Sigma$ be a many-sorted signature and $\Alg$ the set of algebras over $\Sigma$. A \tbf{counterpart model} is a triple $\model:=\angbr{W}{\rightsquigarrow,d}$ such that
\begin{itemize}
    \item $W$ is a set of worlds;
    \item $\freccia{W}{d}{\Alg}$ is a function assigning to each world $\omega\in W$ a $\Sigma$-algebra;
    \item $\rightsquigarrow\subseteq (W\times (\Alg \rightharpoonup \Alg)\times W)$ is the \emph{accessibility relation} over $W$ enriched with  (partial) homomorphisms (\emph{counterpart relations}) between the algebras of the connected worlds.
\end{itemize}
\end{definition}
The elements of $\rightsquigarrow$ are defined such that for every $(w_1 , cr, w_2 ) \in \rightsquigarrow $ we have that $cr : d(w_1 ) \rightharpoonup d(w_2 )$ is a (partial) homomorphism. In particular, each component $cr$ of $\rightsquigarrow$  explicitly defines the counterparts in (the algebras assigned to) the target world of (the algebras assigned to) the source world.

The intuition is that we are considering a transition system labeled with morphisms between
algebras, as an immediate generalization of graph transition systems. The counterpart relations allow to avoid the trans-world identity, i.e. the implicit identification of elements of (the algebras of) different worlds sharing the same name. As a consequence, the name of the elements has a meaning local to the belonging world. For this reason, as we will see, the counterpart
relations allow for the creation, deletion, renaming and merging of elements in a type-respecting way.
\subsection{Relational presheaves models}
The main goal of this section is to explain how the counterpart model introduced in \cite{CounterpartSemanticsGLV} admits a natural generalization and presentation from a categorical setting.


To fix the notation, and in order to provide the intuition behind the models we introduce, we  briefly recall from \cite{GhilardiMeloni1988} an interesting presentation of the notions of \emph{world}, \emph{process} and \emph{individuals} arising from modal logic, in terms of presheaves.

In particular, given a small category $\mW$, its objects $\sigma,\omega, \rho, \dots$ can be considered as \emph{worlds} or \emph{instants of time}, and the arrows $\freccia{\sigma}{f}{\omega}$ of the category represent \emph{temporal developments} or \emph{ways of accessibility}. 
 Notice that in the usual notion of Kripke-frame the accessibility relation $R$ is simply required to be a \emph{binary relation} on the set $W$ of worlds. In particular, this means that two worlds are connected by \emph{at most one temporal development} or that there is \emph{at most one way to pass from a given world to another}. This is a constraint that we would like to avoid. The natural generalization of considering a set of worlds $W$ and a relation $R$ is considering a category $\mW$. 
 
From this perspective, a presheaf $\presheaf{\mW}{D}$ assigns to every world $\omega$ the set $D_{\omega}:=D(\omega)$ of its \emph{individuals}, and to a temporal development $\freccia{\sigma}{f}{\omega}$ a function $\freccia{D_{\omega}}{D_f}{D_{\sigma}}$ between the individuals living in the worlds $\omega$ and $\sigma$. Therefore, if we consider two elements $a\in D_{\omega}$ and $b\in D_{\sigma}$, the equality $b=D_f(a)$ can be read as \emph{a is a future development with respect to $f$ of b}.

In other words, \emph{the notion of presheaf represents the natural categorification of the notion of counterpart frame whose counterpart relation is functional}. In particular, it is direct to check that, given  presheaf $\presheaf{\mW}{D}$, one can define a counterpart frame $\angbr{W,R}{D,d,C}$ as follows
\begin{itemize}
    \item $W$ is given by the objects of the category $\mW$;
    \item $\omega R\sigma $ if and only if there exists at least an arrow $\omega \to \sigma$ of $\mW$;
    \item $d$ and $D$ coincide and are given by the action of the presheaf $\presheaf{\mW}{D}$ on the objects of $\mW$, i.e. $D(\omega)=D(\omega)=D_{\omega}$;
    \item  $C$ assigns to a given 2-tuple $\angbr{\omega}{\omega'}$ the subset of $D(\omega)\times D(\omega')$  whose elements are pairs $\angbr{a}{a'}$ such that $a\in D(\omega)$, $a'\in D(\omega')$ and there exists an arrow $\freccia{\omega}{f}{\omega'}$ such that $a=D_f(a')$.
    
\end{itemize}

Thus, we have seen that the choice of presheaves for the counterpart semantics is quite natural, but it comes with some restrictions: one of the main reasons why such semantics has been introduced is not only to avoid the increasing outer domain condition, but more generally to avoid the constraint that every individual of a world $w$ has to admit a counterpart in every world connected to $w$. Presheaves clearly are not subject to the increasing outer domain condition but if we consider a temporal development $\freccia{\omega}{f}{\sigma}$, and being $D_f$ a total function, we have that for every individual of $D_{\sigma}$ there exists a counterpart in $D_{\omega}$. This forces that an individual $t$ living of the world $\sigma$ \emph{necessarily has a counterpart} in the world $\omega$ with respect to the development $\freccia{\omega}{f}{\sigma}$.  So, to fully abstract the the idea of counterpart semantics in categorical logic and the notion of counterpart frame, we have to consider  the case in which 
\[\relazione{D_{\sigma}}{D_f}{D_{\tau}}\]
is an arbitrary relation. Therefore, in this context considering an element $\angbr{a}{b}\in D_f$ can be read as \emph{$a$ is the future counterpart with respect the development $f$ of $b$.} Therefore, we recall the notion of \emph{relational presheaf}.



\begin{definition}
A \tbf{relational presheaf} is a functor $\relpresheaf{\mC}{X}$, where $\Rel$ denotes the category of sets and relations.
\end{definition}
The generality of the notion of relational presheaf allows for example to deal with \emph{partial functions}, and to avoid the previously described situations. 

Relational presheaves form a category with finite limits if as morphisms we consider the families of set functions $\psi:=\{\freccia{X_{\sigma}}{f_{\sigma}}{U_{\sigma}}\}_{\sigma\in \mC}$ such that for every $\freccia{\sigma}{g}{\tau}$ of the base category, we have that \[\angbr{t}{s}\in X_{g}\Rightarrow\angbr{f_{\tau}(t)}{f_{\sigma}(s)}\in U_{g}\]
where $s\in X_{\sigma}$ and $t\in X_{\tau}$.
We call this kind of morphisms \tbf{relational morphisms}, and we denote the category of relational presheaves and relational morphisms as $\relpresheafcat{\mW}$.

Thus, we have seen the link between counterpart frames and relational presheaves. Now we introduce the new notion of counterpart $\mW$-model and we explain how this notion is a categorification of the notion of counterpart model in the sense of \cite{CounterpartSemanticsGLV}.

\begin{definition}
Let $\Sigma$ be a many-sorted signature. A \tbf{counterpart $\mW$-model} is a triple $\Relmodel=(\mW,\sortmodel,\functionmodel)$ such that
\begin{itemize}
    \item $\mW$ is a category of worlds;
    \item $\sortmodel=\{\relpresheaf{\mW}{\interp{\tau}{\Relmodel}}\}_{\tau\in S_{\Sigma}}$ is a set of relational preshaves on $\mW$, indexed on $S_{\Sigma}$;
    \item $\functionmodel=\{\freccia{\interp{\tau_1}{\Relmodel}\times \cdots \times \interp{\tau_m}{\Relmodel}}{\mI(f_{\Sigma})}{\interp{\tau}{\Relmodel}}\}_{f_{\Sigma}\in F_{\Sigma}}$ is a set of morphisms of relational presheaves.
\end{itemize}
\end{definition}

\begin{definition}
Let us consider a signature $\Sigma$ and a category $\mW$ of worlds. We define the \tbf{category of counterpart $\mW$-models}, denoted by $\ContModelCat{\mW}{\Sigma}$, as the category whose objects are counterpart $\mW$-models $\Relmodel$, and whose morphisms $\freccia{\Relmodel}{F}{\Relmodel'}$ are families of morphisms  $F:=\{\freccia{\interp{\tau}{\Relmodel}}{F_{\tau}}{\interp{\tau}{\Relmodel'}}\}_{\tau\in S_{\Sigma}}$ of relational presheaves, commuting with the relational morphisms of $\functionmodel$ and $\functionmodel'$, i.e. $F_{\tau}\circ \mI(f_{\Sigma})= \mI'(f_{\Sigma})\circ (F_{\tau_1}\times \cdots\times F_{\tau_m})$ for all function symbols  $\freccia{\tau_1\times\cdots\times \tau_m}{f_{\Sigma}}{\tau}\in F_{\Sigma}$ of the signature.
\end{definition}
The notion of  counterpart $\mW$-model admits a clear interpretation also from the categorical perspective of functorial semantics. In details, observe that, by definition, a counterpart $\mW$-model $\Relmodel=(\mW,\sortmodel,\functionmodel)$ assigns to every sort $\tau$ of the signature $\Sigma$ a relational presheaf $\relpresheaf{\mW}{\interp{\tau}{\Relmodel}}$ and to every function symbol $f_{\Sigma}$ a morphism of relational presheaves $\freccia{\interp{\tau_1}{\Relmodel}\times \cdots \times \interp{\tau_m}{\Relmodel}}{\mI(f_{\Sigma})}{\interp{\tau}{\Relmodel}}$. Therefore a counterpart $\mW$-model can be represented as a preserving finite products functor $\freccia{\SynCat{\Sigma}}{F_{\mW}}{\relpresheafcat{\mW}}$ from the syntactic category $\SynCat{\Sigma}$ into the category $\relpresheafcat{\mW}$ of $\mW$-presheaves.

Similarly, every such functor $\freccia{\SynCat{\Sigma}}{F_{\mW}}{\relpresheafcat{\mW}}$ induces a counterpart $\mW$-model. It is thus straightforward to check that the following result holds.

\begin{theorem}
For every category $\mW$, the category $\ContModelCat{\mW}{\Sigma}$ of counterpart $\mW$-models is equivalent to the category $\FPfunctors (\SynCat{\Sigma},\relpresheafcat{\mW})$  of finite product preserving functors and natural transformations from the syntactic category $\SynCat{\Sigma}$ to $\relpresheafcat{\mW}$ of relational presheaves over $\mW$.
\end{theorem}
Notice that if we want to recover in this categorical framework the original idea behind the notion of counterpart models introduced in \cite{CounterpartSemanticsGLV}, i.e. that every world is sent to a $\Sigma$-algebra, we just need to consider counterpart $\mW$-models $\Relmodel=(\mW,\sortmodel,\functionmodel)$ where every relational presheaf of $\sortmodel$ sends a morphism $\freccia{\omega}{f}{\sigma}$ of  $\mW$ to a relation $\relazione{\interp{\tau}{\Relmodel}_{\sigma}}{\interp{\tau}{\Relmodel}_f}{\interp{\tau}{\Relmodel}_{\omega}}$ whose converse $\relazione{\interp{\tau}{\Relmodel}_{\omega}}{(\interp{\tau}{\Relmodel}_f)^{\dagger}}{\interp{\tau}{\Relmodel}_{\sigma}}$ is a partial function. 

In particular, in the following two results we provide the precise link between counterpart models in the sense of \cite{CounterpartSemanticsGLV}
(see Def. \ref{def counterpart model})
and counterpart $\mW$-models. 

\begin{proposition}\label{remark caso particolare counterpar model GLV}
Les us consider a counterpart $\mW$-model $\Relmodel=(\mW,\sortmodel,\functionmodel)$ such that  every relational presheaf of $\sortmodel$ sends a morphism $\freccia{\omega}{f}{\sigma}$ of  $\mW$ to a relation $\relazione{\interp{\tau}{\Relmodel}_{\sigma}}{\interp{\tau}{\Relmodel}_f}{\interp{\tau}{\Relmodel}_{\omega}}$ whose converse $\relazione{\interp{\tau}{\Relmodel}_{\omega}}{(\interp{\tau}{\Relmodel}_f)^{\dagger}}{\interp{\tau}{\Relmodel}_{\sigma}}$ is a partial function. Then the triple $\model_{\Relmodel}:=\angbr{W_{\Relmodel}}{\rightsquigarrow_{\Relmodel},d_{\Relmodel}}$ where

\begin{itemize}
    \item the \emph{set of worlds} is given by the objects of the category $\mW$, i.e. $W_{\Relmodel}=\mathsf{ob}(\mW)$,
    \item $\freccia{W_{\Relmodel}}{d_{\Relmodel}}{\Alg}$ is a function assigning to each world $\omega\in W_{\Relmodel}$ the $\Sigma$-algebra $d_{\Relmodel}(\omega):=(A_{\omega},F^{A_{\omega}}_{\Sigma})$ where $A_{\omega}:=\bigcup_{\tau\in S_{\Sigma}}\{\interp{\tau}{\Relmodel}_{\omega}\}$ and $F^{A_{\omega}}_{\Sigma}:=\bigcup_{f_{\Sigma}\in F_{\Sigma}}\{\mI(f_{\Sigma})_{\omega}\}$,
    \item for every arrow $\freccia{\omega}{f }{\sigma}$ of $\mW$ we define an element $(\omega,cr_f,\sigma)\in \rightsquigarrow_{\Relmodel}$ where the function $\frecciaparziale{d_{\Relmodel}(\omega)}{cr_f}{d_{\Relmodel}(\sigma)}$ is defined on a given typed element $a\in \interp{\tau}{\Relmodel}_{\sigma}$ as  $cr_f(a):=(\interp{\tau}{\Relmodel}_f)^{\dagger}(a)$
\end{itemize}
is a counterpart model.
\end{proposition}

Similarly, one can directly check the dual result.

\begin{proposition}\label{proposition ogni counterpar model ci da un counterpart W model}
Let $\Sigma:=(S_{\Sigma},F_{\Sigma})$ be a many-sorted signature, and let $\model:=\angbr{W}{\rightsquigarrow,d}$ be a counterpart model. Then the triple $\Relmodel_{\model}=(\mW,\sortmodel,\functionmodel)$ where
\begin{itemize}
    \item the category $\mW$ is the category whose objects are the worlds of $W$ and whose arrows are obtained by defining
    for every element $(\omega_1,cr,\omega_2) \in  \rightsquigarrow$ a generating arrow $cr:\omega_1\to\omega_2$,
    \item for a sort $\tau$ of $S_{\Sigma}$ we define the relational presheaf $\relpresheaf{\mW}{\interp{\tau}{\Relmodel_{\model}}}$ by the assignment $\interp{\tau}{\Relmodel_{\model}}_{\omega}:=d(\omega)_{\tau}$ and for a generating arrow $\freccia{\omega_1}{cr}{\omega_2}$ of $\mW$ we define $\interp{\tau}{\Relmodel_{\model}}_{\omega}:=(cr)^{\dagger}_{\tau}$,
    \item for any function symbol $\freccia{\tau_1\times \cdots \times \tau_m}{f_{\Sigma}}{\tau}$ of $\in F_{\Sigma}$ we define the natural transformation $\freccia{\interp{\tau_1}{\Relmodel_{\model}}\times \cdots \times \interp{\tau_m}{\Relmodel_{\model}}}{\mI(f_{\Sigma})}{\interp{\tau}{\Relmodel_{\model}}}$ by the assignment $\mI(f_{\Sigma})_{\omega}:=f_{\Sigma}^{d(\omega)}$
\end{itemize}
is a  counterpart $\mW$-model and every relational presheaf of $\sortmodel$ sends a morphism  $\freccia{\omega}{f}{\sigma}$ of  $\mW$ to a relation $\relazione{\interp{\tau}{\Relmodel_{\model}}_{\sigma}}{\interp{\tau}{\Relmodel_{\model}}_f}{\interp{\tau}{\Relmodel_{\model}}_{\omega}}$ whose converse $\relazione{\interp{\tau}{\Relmodel_{\model}}_{\omega}}{(\interp{\tau}{\Relmodel_{\model}}_f)^{\dagger}}{\interp{\tau}{\Relmodel_{\model}}_{\sigma}}$ is a partial function.
\end{proposition}
\begin{remark}
Let $\Sigma:=(S_{\Sigma},F_{\Sigma})$ be a many-sorted signature, and let $\model:=\angbr{W}{\rightsquigarrow,d}$ be a counterpart model. Notice that if we first construct the counterpart $\mW$-model $\Relmodel_{\model}$ employing Proposition~\ref{proposition ogni counterpar model ci da un counterpart W model}, and then we construct the counterpart model $\model_{(\Relmodel_{\model})}$ employing Proposition~\ref{remark caso particolare counterpar model GLV}, we have that the counterpart model $\model_{(\Relmodel_{\model})}$ may be different from
$\model$. As it will be shown in Theorem~\ref{teorema corrispondenza 1-1},
when considering \emph{temporal} structures such a difference will be irrelevant, from a semantic point of view.
\end{remark}

\begin{example}\label{esempio relational pref. su grafi}
Let us consider the ordinary signature $\Sigma_{Gr}=(S_{Gr},F_{Gr})$ for directed graphs introduced in Example \ref{esempio sigma algebra grafi}.
The notion of counterpart $\mW$-model allows us to provide a categorical presentation of the counterpart semantics  presented in \cite{CounterpartSemanticsGLV}. 
In this case, a counterpart $\mW$-model $\Relmodel_{Gr}=(\mW,\mathfrak{S}_{\Sigma_{Gr}},\mathfrak{F}_{\Sigma_{Gr}})$ consists of
\begin{itemize}
    \item a category of worlds $\mW$;
    \item $\mathfrak{S}_{\Sigma_{Gr}}=\{\interp{\tau_N}{\Relmodel},\interp{\tau_E}{\Relmodel}\}$, where $ \interp{\tau_N}{\Relmodel}$and $\interp{\tau_E}{\Relmodel}$ are two relational presheaves on the category $\mW$ of worlds such that $(\interp{\tau_N}{\Relmodel}_f)^{\dagger}$ and $(\interp{\tau_E}{\Relmodel}_f)^{\dagger}$ are partial functions for every arrow $\freccia{\omega}{f}{\sigma}$ of $\mW$;
    \item $\mathfrak{F}_{\Sigma_{Gr}}=\{\mI(s),\mI(t)\}$, where$\freccia{\interp{\tau_E}{\Relmodel}}{\mI(s)}{\interp{\tau_N}{\Relmodel}}$ and $\freccia{\interp{\tau_E}{\Relmodel}}{\mI(t)}{\interp{\tau_N}{\Relmodel}}$ are two morphisms of relational presheaves such that every component $\mI(t)_{\omega}$ and $\mI(s)_{\omega}$ is a function.
    \end{itemize}
As anticipated in Proposition \ref{remark caso particolare counterpar model GLV}, to present as particular case the notion of counterpart model  in~\cite{CounterpartSemanticsGLV}, where each world is sent to a directed graph, we needed to require that both $\mI(t)_{\omega}$ and $\mI(s)_{\omega}$ are functions for every world $\omega$.
Therefore, given the model $\Relmodel_{Gr}$, we have again that every world $\omega$ is mapped to a directed graph \[d(\omega):=(\interp{\tau_N}{\Relmodel}_{\omega},\interp{\tau_E}{\Relmodel}_{\omega}, \mI(s)_{\omega},\mI(t)_{\omega}) \] identified by the set of nodes $\interp{\tau_N}{\Relmodel}_{\omega}$, the set of arcs $\interp{\tau_E}{\Relmodel}_{\omega}$, and the two functions $\freccia{\interp{\tau_E}{\Relmodel}_{\omega}}{\mI(s)_{\omega},\mI(t)_{\omega}}{\interp{\tau_N}{\Relmodel}_{\omega}}$. Moreover, every morphism $\freccia{\omega}{f}{\sigma}$ of the category $\mW$ induces a \emph{partial homomorphism} of directed graphs \[\frecciaparziale{d(\omega) }{cr_f}{d(\sigma) }\]
given by $cr_f:=((\interp{\tau_N}{\Relmodel}_f)^{\dagger},(\interp{\tau_E}{\Relmodel}_f)^{\dagger})$.
\end{example}

\begin{example}
As an instance of Example~\ref{esempio relational pref. su grafi}, we consider our running example, i.e. the three graphs $\mathbf{G}_0$, $\mathbf{G}_1$ and $\mathbf{G}_2$ presented in Fig. \ref{Figura grafi } of Example \ref{esempio sigma algebra grafi}. In this case we consider a category $\mW$ whose objects are three worlds $\omega_0$, $\omega_1$ and $\omega_2$ and the morphisms are generated by the compositions of
\[
\xymatrix{
\omega_0\ar[r]^{f_0} & \omega_1  \ar@<-.5ex>[r]_{f_2} \ar@<.5ex>[r]^{f_1} & \omega_2\ar[r]^{f_3}& \omega_2.
}
\]
 The relational presheaves of nodes and edges are given by the assignment
 \begin{itemize}
   \setlength\itemsep{0.5em}
     \item $\interp{\tau_N}{\Relmodel}_{\omega_i}:=N_{\mathbf{G}_i}$, where $N_{\mathbf{G}_i}$ is the set of nodes of the graph $\mathbf{G}_i$;
     \item $\interp{\tau_N}{\Relmodel}_{f_0}:=\{(n_3,n_0),(n_4,n_1),(n_3,n_2) \}\subseteq N_{\mathbf{G}_1}\times N_{\mathbf{G}_0}$;
     \item $\interp{\tau_N}{\Relmodel}_{f_1}:=\{(n_5,n_3),(n_5,n_4)\}\subseteq N_{\mathbf{G}_2}\times N_{\mathbf{G}_1}$;
     \item $\interp{\tau_N}{\Relmodel}_{f_2}:=\{(n_5,n_3),(n_5,n_4)\}\subseteq N_{\mathbf{G}_2}\times N_{\mathbf{G}_1}$;
     \item $\interp{\tau_N}{\Relmodel}_{f_3}:=\{(n_5,n_5)\}\subseteq N_{\mathbf{G}_2}\times N_{\mathbf{G}_2}$;
 \end{itemize} and 
\begin{itemize}
  \setlength\itemsep{0.5em}
    \item $\interp{\tau_E}{\Relmodel}_{\omega_i}:=E_{\mathbf{G}_i}$, where $E_{\mathbf{G}_i}$ is the set of edges of the graph $\mathbf{G}_i$.
      \item $\interp{\tau_E}{\Relmodel}_{f_0}:=\{(e_3,e_0),(e_4,e_1) \}\subseteq E_{\mathbf{G}_1}\times E_{\mathbf{G}_0}$;
     \item $\interp{\tau_E}{\Relmodel}_{f_1}:=\{(e_5,e_3)\}\subseteq E_{\mathbf{G}_2}\times E_{\mathbf{G}_1}$;
      \item $\interp{\tau_E}{\Relmodel}_{f_2}:=\{(e_5,e_4)\}\subseteq E_{\mathbf{G}_2}\times E_{\mathbf{G}_1}$;
       \item $\interp{\tau_E}{\Relmodel}_{f_3}:=\{(e_5,e_5)\}\subseteq E_{\mathbf{G}_2}\times E_{\mathbf{G}_2}$.
\end{itemize} 
The natural transformations $\mI(s),\mI(t)$ are the domain and codomain maps. Notice that both $(\interp{\tau_N}{\Relmodel}_{f_i})^{\dagger}$ and $(\interp{\tau_E}{\Relmodel}_{f_i})^{\dagger}$ are partial functions for $i=0\dots 3$. Therefore we are under the hypotheses of Proposition~\ref{remark caso particolare counterpar model GLV}. Thus, following the assignment of Example~\ref{esempio relational pref. su grafi}, we can consider the counterpart model $\angbr{W}{\rightsquigarrow,d}$ corresponding to the counterpart $\mW$-model, and we have that $d(\omega_i)=\mathbf{G}_i$, and $cr_{f_i}$ can be represented graphically as follows

\begin{figure}[H]\label{Figura grafi e morfismi parziali}
\centering
\begin{tikzpicture}[scale=0.40]]
	\begin{pgfonlayer}{nodelayer}
		\node [style=new style 0] (0) at (-11, 6) {$n_0$};
		\node [style=new style 0] (1) at (-5, 6) {$n_1$};
		\node [style=new style 0] (2) at (-8, 2) {$n_2$};
		\node [style=new style 1] (3) at (-11, 4) {$e_2$};
		\node [style=new style 1] (4) at (-5, 4) {$e_1$};
		\node [style=new style 1] (5) at (-8, 8) {$e_0$};
		\node [style=new style 0] (6) at (-1, 5) {$n_3$};
		\node [style=new style 0] (7) at (4.75, 5) {$n_4$};
		\node [style=new style 1] (8) at (2, 8) {$e_3$};
		\node [style=new style 1] (9) at (2, 2) {$e_4$};
		\node [style=new style 0] (10) at (11, 6.25) {$n_5$};
		\node [style=new style 1] (13) at (11, 3.5) {$e_5$};
	\end{pgfonlayer}
	\begin{pgfonlayer}{edgelayer}
		\draw (2) to (3);
		\draw (0) to (5);
		\draw (1) to (4);
		\draw [style=freccia] (4) to (2);
		\draw [style=freccia] (3) to (0);
		\draw [style=freccia] (5) to (1);
		\draw (6) to (8);
		\draw (7) to (9);
		\draw [style=freccia] (9) to (6);
		\draw [style=freccia] (8) to (7);
		\draw [bend right=60, looseness=1.50] (10) to (13);
		\draw [style=freccia, bend right=60, looseness=1.50] (13) to (10);
		\draw [style={freccia_tratteggiata}, bend left=15] (5) to (8);
		\draw [style={freccia_tratteggiata}, bend left=15] (1) to (7);
		\draw [style={freccia_tratteggiata}, bend right=15] (0) to (6);
		\draw [style={freccia_tratteggiata}, bend right=45] (2) to (6);
		\draw [style={freccia_tratteggiata}, bend right, looseness=0.75] (4) to (9);
		\draw [style={freccia_tratteggiata}, in=-120, out=-30, loop] (13) to ();
		\draw [style={freccia_tratteggiata}, in=135, out=45, loop] (10) to ();
		\draw [style=new edge style 0, bend right=45] (9) to (13);
		\draw [style=new edge style 0, bend right=45] (6) to (10);
		\draw [style=new edge style 0, bend right] (7) to (10);
		\draw [style=new edge style 1, bend left] (7) to (10);
		\draw [style=new edge style 1, bend left] (6) to (10);
		\draw [style=new edge style 1, in=-165, out=0, looseness=1.25] (8) to (13);
	\end{pgfonlayer}
\end{tikzpicture}
\caption{A counterpart model with three sequential worlds}
\end{figure}
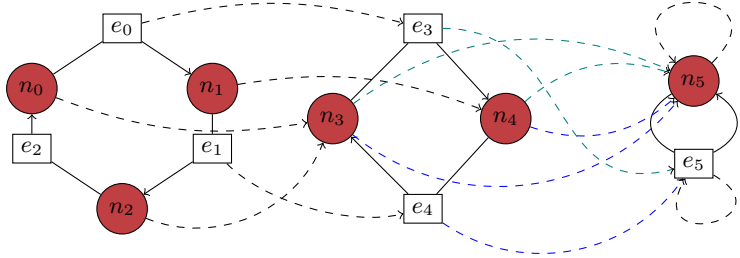
This counterpart model in Figure~\ref{Figura grafi e morfismi parziali} illustrates two  executions of our running example instantiated with three processes (edges). The two evolutions try to instantiate the main features of our logic, as presented in the next sections. The counterpart relations (drawn with dotted lines and colors to distinguish $cr_{f_1}$ from $cr_{f_2}$) reflect
the fact that at each transition one process (edge) is discarded and its source and target channels (nodes)
are merged. In particular, the transition
$cr_{f_0}:=((\interp{\tau_N}{\Relmodel}_{f_0})^{\dagger},(\interp{\tau_E}{\Relmodel}_{f_0})^{\dagger})$ deletes edge $e_2$ and merges nodes $n_0$ and $n_2$ into $n_3$. Similarly for $cr_{f_1}$ and $cr_{f_2}$, while $cr_{f_3}$ is a cycle preserving both $e_5$ and $n_5$ , denoting that the system is idle, yet alive.
\end{example}
\subsection{Relational power-set presheaf}
While presheaves form a topos, and then the category of presheaves is strong enough to deal with higher-order features, relational presheaves on a given category rarely have such a structure, as observed in \cite{Niefield2010}. 
This is due to the fact that the category $\Set$ is a topos while the category of sets and relations $\Rel$ is not a topos, but an \emph{allegory} \cite{freyd1990categories}.
To give an intuition of what an allegory is we recall the paradigm of \cite{freyd1990categories}: \emph{allegories  are to binary relations between sets as categories are to functions between sets}. 

However, even if  the category of sets and relations lacks the topos structure, and in particular \emph{power-objects}, one could employ relational presheaves for higher-order using the structure of \emph{power-allegory} of the category of relations. For the formal definition and the proof that $\Rel$ is a power-allegory we refer the reader to \cite[Prop. 2.414]{freyd1990categories}, while now we briefly discuss how one can define the \emph{power-set relational presheaf} $\relpresheaf{\mC}{\powerelpresheaf{X}}$ of a given presheaf $\relpresheaf{\mC}{X}$.

So, let us consider a relational presheaf $\relpresheaf{\mC}{X}$. While the action of the functor $\relpresheaf{\mC}{\powerelpresheaf{X}}$ we are going to define is clear on the objects, i.e. $\sigma$ is sent to the power-set $\mathrm{P}(X_{\sigma} )$ of $X_{\sigma}$, the problem of lifting a relation $\relazione{X}{R}{Y}$ to a relation on the power-sets $\relazione{\mathrm{P}(A)}{R'}{\mathrm{P}(B)}$ has to be treated carefully.

In fact, one could define for every morphism $\freccia{\sigma}{f}{\omega}$  of $\mC$ the relation $\relazione{\powerelpresheaf{X}_{\omega}}{\powerelpresheaf{X}_f}{\powerelpresheaf{X}_{\sigma}}$ such that for every $A\subseteq X_{\omega}$ and $B\subseteq X_{\sigma}$ we have
    \[\angbr{A}{B}\in\powerelpresheaf{X}_f \iff (\forall a\in A, \; \exists b\in B : \angbr{a}{b}\in X_f)\And (\forall b\in B, \; \exists a\in A : \angbr{a}{b}\in X_f). \]

This is a quite natural way of lifting relations to relations of power-sets, but notice that in this case $\powerelpresheaf{X}$ is just a lax-functor, and not a functor. In fact, we have that $\powerelpresheaf{X}_{gf}\subseteq \powerelpresheaf{X}_f\powerelpresheaf{X}_g$, but these two sets are not equal. 

Since we are interested in considering relational presheaves in the strict sense, we have to lift relations in a different way. To this aim we employ the known equivalence between relations and Galois connections (or maps) on power-sets \cite{GARDINER199421}, i.e. the equivalence between $\Rel$ and $ \mathrm{Map}(\mathbf{Pow}) $.

Recall from \cite[Ex. 2]{GARDINER199421} that given a relation $\relazione{A}{R}{B}$, we can define a function (preserving arbitrary unions) $\freccia{\mathrm{P}(B)}{P_R}{\mathrm{P}(A)}$ by assigning 
$$P_R(S):=\{ a\in A \; | \exists b\in S \; : aRb\}$$
to every subset $S\subseteq B$. Moreover, given the equivalence $\Rel\equiv \mathrm{Map}(\mathbf{Pow}) $, and using the fact that $\mathbf{Pow}=(\mathbf{Pow}^{\op})^{\op}$, we can immediately conclude that the assignment $R\mapsto P_R $ preserves the compositions and the identities.

\begin{definition}
\label{definizione relational powerset}
Let $\relpresheaf{\mC}{X}$ be a relational presheaf. The \tbf{relational power-set presheaf} $\relpresheaf{\mC}{\powerelpresheaf{X}}$ is the functor defined as
\begin{itemize}
    \item for every object $\sigma\in \mC$, $\powerelpresheaf{X}(\sigma)$ is the power-set of $X_{\sigma}$;
    \item for every $\freccia{\sigma}{f}{\omega}$ we define the relation $\relazione{\powerelpresheaf{X}_{\omega}}{\powerelpresheaf{X}_f}{\powerelpresheaf{X}_{\sigma}}$ as $\powerelpresheaf{X}_f:=P_{X_f}$.
\end{itemize}
\end{definition}

In particular, the relational presheaf $\relpresheaf{\mC}{\powerelpresheaf{X}}$ is an ordinary presheaf $\presheaf{\mC}{\powerelpresheaf{X}}$. Finally, given a relational presheaf $\relpresheaf{\mC}{X}$ we define the \emph{epsiloff relational presheaf} $\relpresheaf{\mC}{\in_{X}}$. 
\begin{definition}
Let $\relpresheaf{\mC}{X}$ be a relational presheaf. The \tbf{epsiloff relational presheaf} is the functor $\relpresheaf{\mC}{\in_{X}}$ defined as
\begin{itemize}
    \item for every $\sigma\in \mC$, $\in_X(\sigma):=\{(a,A)\in X_{\sigma} \times \powerelpresheaf{X}_{\sigma} \; |\; a\in A\}$;
    \item for every $\freccia{\sigma}{f}{\omega}$, $(\in_X)_f$ is the relation given by $\angbr{(b,B)}{(a,A)}\in(\in_X)_f $ if and only if $\angbr{b}{a}\in X_f$ and $\powerelpresheaf{X}_f(B)=A$ where $B\subseteq X_{\omega}$ and $A\subseteq X_{\sigma}$.
\end{itemize}
\end{definition}

It is direct to check that $\in_X$ is a functor, i.e. that it preserves compositions and identities, and also that, while the relational power-set presheaf is a set-valued sheaf, the epsiloff presheaf is just a relational presheaf in general.

\section{Syntax of quantified temporal logic}
Before presenting the syntax of our logic, we consider a set of second-order variables $\chi\in \mathcal{X}$, where a variable $\chi_{\tau}$ with sort $\tau\in S_{\Sigma}$ ranges over sets of elements of sort $\tau$. 


\begin{definition}
Let $\Sigma$ be a many-sorted signature, $X$ a set of first-order variables and $\mathcal{X}$ a set of second-order variables, both typed over $S_{\Sigma}$. The set $\mF_{\Sigma}$ of formulae of our temporal logic is generated by the rules

\[\phi:= \mathsf{tt} \; | \; \varepsilon\in_{\tau} \chi  \;|\; \neg \phi\; |\; \phi \vee \phi \;| \; \exists_{\tau}x. \phi \;| \; \exists_{\tau}\chi. \phi \; |\; \nextoperator \phi \; |\; \phi_1 \untiloperator \phi_2.  \]
Whenever clear from the context, subscripts and types may be omitted.
\end{definition}

The \emph{next} operator $\nextoperator$ provides a way of asserting that something has to be true at the next step, i.e. $\nextoperator \phi$ means that $\phi$ has to hold at the next state. 
The \emph{until} operator $\untiloperator$ can be explained as follows: $ \phi_1 \untiloperator \phi_2$ means that $\phi_1$ has to hold at least until $\phi_2$ becomes true, which must hold at the current or a future position. 
The \emph{sometimes} modality $\dm$ is obtained as $\dm \phi:=\mathsf{tt} \untiloperator \phi$ and the \emph{always} modality $\sq$ as $\sq \phi:= \neg \dm \neg \phi$.

 Notice that the standard boolean connectives $\vee,\rightarrow, \leftrightarrow$ and the universal quantifiers can be derived as usual. Moreover, the typed equality $=_{\tau}$ can be derived as $\epsilon_1=_{\tau}\epsilon_2\equiv \forall_{\tau} \chi.( \epsilon_1\in_{\tau} \chi \leftrightarrow \epsilon_2\in_{\tau} \chi)$.
We also remark that $\in_{\tau}$ is a family of membership predicates typed over $S_{\Sigma}$ indicating that the
evaluation of a term with sort $\tau$ belongs to the evaluation of a second-order variable with the same sort
$\tau$.


As usual, we consider formulae in context, defined as $[\Gamma, \Delta] \; \phi$, where $\phi$ is a formula of $\mF_{\Sigma}$, $\Gamma$ is a first-order context and $\Delta$ is the second-order context.

We now move to a different syntax that is equivalent to the previous one, yet  describes only those
formulae which are in \emph{positive form}. Such a syntax includes derived operators in order to have
negation applied only to atomic proposition. This is a standard presentation for linear time temporal logics, and it will be pivotal for our semantics, since it allows to capture negation in terms of set complementation.

\begin{definition}
Let $\Sigma$ be a many-sorted signature, $X$ a set of first-order variables and $\mathcal{X}$ a set of second-order variables, both typed over $S_{\Sigma}$. The set $\mF_{\Sigma}$ of positive formulae of our temporal logic is generated by the rules

\[\phi:= \mathsf{tt} \; | \; 
\varepsilon\in_{\tau} \chi  \;|\; \neg \phi
\]
and
\[\psi:= \phi\;|\;\psi \vee \psi \;| \; \psi \wedge \psi \;| \;  \exists_{\tau}x. \phi \;| \; \exists_{\tau}\chi. \psi \; |\;  \;| \;  \forall_{\tau}x. \phi \; |\;\forall_{\tau}\chi. \psi \; |\; \nextoperator \psi \; |\; \psi_1 \untiloperator \psi_2  \; |\; \psi_1 \wuntiloperator \psi_2  \]
where  $\psi_1 \wuntiloperator \psi_2 $ denotes the \emph{weak until operator}.
\end{definition}

Now $ \psi_1 \wuntiloperator \psi_2$ means that $\psi_1$ has to hold at least until $\psi_2$ or, if $\psi_2$ never becomes true, $\psi_1$ must remain true forever. We will often use $\mathsf{ff}$ for $\neg \mathsf{tt}$.
In this setting, the operators $\dm$ and $\sq$ can be presented as $\dm \psi := \mathsf{tt}\untiloperator \psi$ and $\sq\psi := \psi \wuntiloperator \mathsf{ff}$.

On the following we will mostly consider just positive formulae, hence dropping the adjective whenever it is not ambiguous.

\begin{example}
\label{example prop. sintat.}
Consider again the graph signature, the counterpart model of Figure \ref{Figura grafi e morfismi parziali}, and the typed predicates $\mathbf{present}_{\tau} (x)\equiv \exists_{\tau} y.x=_{\tau} y$ regarding the presence of an entity with sort $\tau$ in a world (the typing is usually omitted). Combining this with next operator we can speak about elements that are present at the given world and that will be present at the next step for example, i.e.  $\mathbf{present}_{\tau} (x)\wedge \nextoperator ( \mathbf{present}_{\tau} (x))$.

Instead, notice that even if $\neg \mathbf{present}_{\tau} (x)$ is a formula of our logic, it is not in positive form. An equivalent formula in positive form is $\forall_{\tau} y.x\neq_{\tau} y$, where the expression $x\neq_{\tau} y$ is equivalent to 
$\exists_{\tau} \chi.( x\in_{\tau} \chi \wedge y \not \in_{\tau} \chi)$. For the sake of conciseness, we will freely use both $\neg \mathbf{present}_{\tau} (x)$ and $x\neq_{\tau} y$.

Moreover, consider the predicate $\mathbf{loop}(x) \equiv s(x) = t(x)$ characterizing the presence of a cycle. The following formulae
express different properties of our running example: $\nextoperator ( \exists_{\tau_E} x.\mathbf{loop}(x))$ states that at the next step there exists a loop, i.e. an edge whose source and target are equal. The formula $\sq (\mathbf{present}_{\tau_E} (x))$ means that for all the evolutions of our systems there exists at least an edge. The formula $\exists_{\tau_N}x. (((x\neq y) \wedge\nextoperator (x=y))$ means that given a node $y$, there exists another different node that at the next step will be identified with $y$ at the next step. Finally, the formula $ \sq (\exists_{\tau_E}e. s(e)=x\wedge t(e)=y) $ means that the nodes $x$ and $y$ will always be connected by an edge. 

\end{example}
\section{Temporal structures and semantics}

\subsection{Temporal structures}
The notion of \emph{hyperdoctrine} was introduced by Lawvere in a series of seminal papers \cite{Lawvere1969b,Lawvere1969} to provide a categorical framework for first order logic. 

In recent years, this notion has been generalized in several settings, for example introducing elementary and existential doctrines~\cite{QCFF,ECRT} and 
modal hyperdoctrines~\cite{Braner2007FirstorderML}.
%
In particular, in \cite{GhilardiMeloni1996,GhilardiMeloni1990} modal hyperdoctrines are introduced employing the notion of \emph{attribute} associated to a given presheaf, or more generally, to a relational presheaf as a categorification of the model semantics. We start recalling the notion of set of classical attributes, and then we show how by simply fixing a class of morphisms of the base category of a presheaf, we can construct the main temporal operators.

\begin{definition}
Let $\relpresheaf{\mW}{X}$ be a relational presheaf. We define the set
\[ \att(X):=\{ \{A_{\omega}\}_{\omega\in \mW} \; |\; A_{\omega}\subseteq X_{\omega}\}\] 
whose objects are called \tbf{classical attributes}.
\end{definition}
The set of classical attributes has a natural structure of complete boolean algebra when we consider the order provided by the inclusion.
Moreover, a morphism $\freccia{X}{f}{D}$ of relational presheaves induces a morphism of boolean algebras  $$\freccia{\att(D)}{f^*}{\att(X)}$$
between the sets of classical attributes. In particular, this action is given by pulling back world-by-world, i.e. for $A\in \att (D)$, computing the pullback for every world $\omega$ as below and defining $f^*(A):=\{f_{\omega}^*(A_{\omega})\}_{\omega\in  \mW}\in \att (X)$ for each $A\in \att( D)$.

\[\xymatrix{
f_{\omega}^*(A_{\omega})\ar@{^{(}->}[r]\ar[d] & X_{\omega}\ar[d]^{f_\omega}\\
A_{\omega}\ar@{^{(}->}[r]& D_{\omega}
}\]

Given a relational presheaf $\relpresheaf{\mW}{X}$ and a classical attribute $A=\{A_{\omega}\}_{\omega\in \mW}$ of $ \att(X)$ and an element $s\in X_{\omega}$, we use the validity notation 
$s\vDash^X_{\omega}A$
to mean that $s\in A_{\omega}$. Therefore, writing $s\vDash^X_{\omega}A$ means that at the world or instant $\omega$, an individual $s$ of $X_{\omega}$ satisfies the property $A$.
 
\begin{definition}
We say that a relational presheaf  $\relpresheaf{\mW}{X}$ is \tbf{equipped with a temporal structure} $\mathrm{T}$, if the base category $\mW$ is considered together with a class of morphisms $\mathrm{T}$.
\end{definition}
Given a relational presheaf $\relpresheaf{\mW}{X}$, the idea is that the class $\mathrm{T}$ represents the \emph{atomic processes} or the \emph{indecomposable operations} of $\mW$. 
\begin{definition}
Given a temporal structure $\mathrm{T}$ on $\relpresheaf{\mW}{X}$ we denote by $\cammini{\mathrm{T}}{\omega}$ the class of sequences $\mathrm{t}:=(t_1,t_2,t_3,\dots)$ of arrows such that $t_n\in T$ and such that $\domain{t_1}=\omega$ and $\codomain{t_i}=\domain{t_{i+1}}$ for $i\geq 1$.
\end{definition}

 Given a sequence $\mathrm{t}:=(t_1,t_2,t_3,\dots)$ we denote by $\mathrm{t}_{\leq i}$ the arrow $t_it_{i-1}\cdots t_1$. Moreover we denote $\omega_i:=\codomain{t_i}$.
 The intuition is that the class of arrows $\cammini{\mathrm{T}}{\omega}$ represents the $\mathrm{T}$\emph{-evolutions of} the state $\omega$.
The choice of the name \emph{temporal structure} is due to the fact that by simply fixing a class $\mathrm{T}$, we can define operators $\nextoperator$, $\untiloperator$, $\sq$ and $\dm$ on the complete boolean algebra $\att(X)$ of classical attributes. 
\begin{definition}\label{definizione temporal operators prefasci relazionali}
Let us consider a temporal structure $\mathrm{T}$ on $\relpresheaf{\mW}{X}$. For every element $A:=\{A_{\omega}\}_{\omega\in \mW} $ of $\att (X)$  and $s\in X_{\omega} $ we define

\begin{itemize}
   \setlength\itemsep{0.5em}
    \item $s \vDash^X_{\omega} \nextoperator (A) $ if and only if for every arrow $\freccia{\omega}{t}{\sigma}$ of $\mathrm{T}$, there exists an element $z\in D_{\sigma} $ such that $\angbr{z}{s}\in X_t$ and  $z \vDash^X_{\sigma} A $; 
    \item $s\vDash^X_{\omega} A\untiloperator B$ if for every $\mathrm{t}\in \cammini{T}{\omega}$ there exists an $\bar{n}$ such that for every $i\leq \bar{n}$ there exists $z_i\in X_{\omega_i} $ such that $\angbr{z_i}{s}\in X_{t_{\leq i}}$ and $z_i \vDash^X_{\omega_i} A $ and an element $z_{\bar{n}}\in X_{\omega_{\bar{n}}} $ such that $\angbr{z_{\bar{n}}}{s}\in X_{t_{\leq \bar{n}}}$ and $z_{\bar{n}} \vDash^X_{\omega_{\bar{n}}} B $.
   \item $s\vDash^X_{\omega} A\wuntiloperator B$ if for every $\mathrm{t}\in \cammini{T}{\omega}$, we have that for every $i$ there exists $z_i\in X_{\omega_i} $ such that $\angbr{z_i}{s}\in X_{t_{\leq i}}$ and $z_i \vDash^X_{\omega_i} A $ or there exists an $\bar{n}$ such that for every $i\leq \bar{n}$ there exists $z_i\in X_{\omega_i} $ such that $\angbr{z_i}{s}\in X_{t_{\leq i}}$ and $z_i \vDash^X_{\omega_i} A $ and an element $z_{\bar{n}}\in X_{\omega_{\bar{n}}} $ such that $\angbr{z_{\bar{n}}}{s}\in X_{t_{\leq \bar{n}}}$ and $z_{\bar{n}} \vDash^X_{\omega_{\bar{n}}} B $.
   \end{itemize}
   \end{definition}
Recall that given a relational presaheaf $\relpresheaf{\mW}{X}$, the top element of the boolean algebra $\att (X)$ is given by the attributes $\top=\{ X_{\omega}\}_{\omega\in \mW}$ and the bottom element by $\bot=\{ \emptyset_{\omega}\}_{\omega\in \mW}$ since the order of $\att (X)$ is given by the inclusion of sets.
Therefore, employing the previous notions, we can define the operators $\dm A$ and $\sq A$ as $\dm A := \top\untiloperator A$ and $\sq A := A \wuntiloperator \bot$.
\begin{remark}
Notice also that by definition, we have that  $s\vDash^X_{\omega} A\untiloperator B$ iff $s\in C_{\omega}$ for every   $C_{\omega}\subseteq X_{\omega}$ such that $C_{\omega}=B_{\omega} \vee (A_{\omega} \wedge \nextoperator (C)_{\omega})$
(and similarly for $s\vDash^X_{\omega} A\wuntiloperator B$). In other words, the set $(A\untiloperator B)_{\omega}$ can be more concisely described as the least fixed point, with respect to the order given by set inclusion,  of the function $B_{\omega} \vee (A_{\omega} \wedge \nextoperator (-)_{\omega})$ 
and analogously, $(A\wuntiloperator B)_{\omega}$ as the greatest fixed point of the same function
(recall that the boolean structures of attributes is given by the set theoretic inclusions, hence $\vee$ and $\wedge$ are given by the union and the intersection of sets, respectively).

\end{remark}
\begin{example}
Let us consider a temporal structure $\mathrm{T}$ on $\relpresheaf{\mW}{X}$. Given an attribute $A\in \att (X)$ and $s\in X_{\omega}$ it is direct to check that $\dm A$ and $\sq A $ can be directly described as follows
   \begin{itemize}
      \setlength\itemsep{0.5em}
    \item $s\vDash^X_{\omega}\sq A$  if for every $\mathrm{t}\in \cammini{T}{\omega}$ and for every $i$ there exists an element $z_i\in X_{\omega_i} $ such that $\angbr{z_i}{s}\in X_{t_{\leq i}}$ and $z_i \vDash^X_{\omega_i} A $;
    \item $s\vDash_{\omega}^X \dm A$ if for every $\mathrm{t}\in \cammini{T}{\omega}$ there exist $\bar{n}$ and an element  $z_{\bar{n}}\in X_{\omega_{\bar{n}}} $ such that $\angbr{z_{\bar{n}}}{s}\in X_{t_{\leq {\bar{n}}}}$ and  $z_{\bar{n}} \vDash^X_{\omega_{\bar{n}}} A $.
\end{itemize}
\end{example}
\begin{remark}

Notice that in the definition of classical attributes, the relational structure of a given relational presheaves has no rule, while it is fundamental in defining the temporal operators embodying the fundamental idea of the counterpart semantics. 

However, depending of what kind of logic we are interested in, different conditions can be imposed or required on the notion of attribute. For example, one can impose the so called \emph{stability condition with respect to the past}, i.e. for every relational presheaf $\relpresheaf{\mC}{X}$, we define 
\[ \relatt(X):=\{ \{A_{\omega}\}_{\omega\in \mC} \; |\; A_{\omega}\subseteq X_{\omega}\}\] 
where every family $\{A_{\omega}\}_{\omega\in \mC}$ satisfies the following stability condition 
\[ \freccia{\sigma}{f}{\tau}, \; \angbr{t}{s}\in X_f \mbox{ and } t\in A_{\tau} \mbox{ implies } s\in A_{\sigma}.\]

This kind of attributes are introduced in \cite{GhilardiMeloni1996} to apply the relational semantics in intuitionistic logic. However, considering this notion of attributes requires a careful analysis of the interpretation of quantifiers because, for example, the usual Beck Chevalley condition in general is not satisfied in this semantics. For details we refer to~\cite{GhilardiMeloni1996}.
\end{remark}

\begin{example}
\label{example LTL}
A \emph{flow of time} is a pair $(T,<)$ where $T$ is a non-empty set whose elements are called \emph{time points} and $<$ is an irreﬂexive and transitive binary relation on $T$. Given two time points $s,u\in T$, $s<u$ intuitively means that $s$ is \emph{earlier} of $u$. 
Requiring that $<$ is irreﬂexive is thus motivated by the fact that no time point should be in the past or future of itself, and $<$ has to be transitive because $t$ is earlier than $u$, and $u$ earlier than $v$, then we expect $t$ to be earlier than $v$.
A flow of time is \emph{linear} if given any two distinct time points in it, one is before the other, i.e. for every $s,t\in T$ such that $s\neq t$ we have that either $s<t$ or $t<s$. 
A natural choice of base and temporal structure is thus a linear flow of time $(T,\leq)$ and the free category $\mW(T)$ on it. Then every relational presheaf $\relpresheaf{\mW(T)}{X}$  admits a natural temporal structure $\mathrm{T}$ provided by the set of arrows of relation $<$. We will observe in the following section that this choice allows us to obtain the standard semantics for LTL as particular case of our constructions and notions.  In particular, the interpretation of the temporal operators $\nextoperator,\dm,\sq, \untiloperator,\wuntiloperator$ coincide precisely with the usual interpretation of LTL.
\end{example}
\begin{example}
The notions of Kripke frame and of flow of time are an intuitive abstraction of systems that evolve or change over instants of time. However, notice that both of them suffer of the limitation that the relation  between worlds or instants of time has to be binary. This represents a serious constraint if, for example, one want to consider a simple situation in which a system could evolve in another one through two distinct developments. Moreover, the requirement that the relation $<$ in the definition of flow of time has to be transitive and irreflexive does not allow to consider reversible process or reversible developments. This is quite natural if one has to work with time developments, but if we are interesting in a more general situation, we have to consider different notions, such as e.g. labelled transition systems. Recall that a \emph{labelled transition system} (LTS) is a triple $(S,L,\rightarrow)$, where $S$ is a non-empty set of \emph{states}, $L$ is the set of \emph{labels}, and $\rightarrow\subset S\times L\times S$ is a relation, which is total in the first component, i.e. for every $s\in S$ there exists a label $l\in L$ and a state $s'$ such that $(s,l,s')\in \rightarrow$. Therefore, another meaningful choice for the category of worlds and the temporal structure is considering an LTS $(S,L,\rightarrow)$ and the free category $\mW(S)$ on it. Then every relational presheaf $\relpresheaf{\mW(S)}{X}$ admits a temporal structure give by the set $\mathrm{T}_S$ of arrows of the relation $\rightarrow$, where the next time operator $\nextoperator$ has exactly the intuitive meaning of \emph{it holds at every next-step of length one}.
\end{example}
\subsection{Semantics via temporal structures}
In this section we show how relational presheaves and temporal structures can be employed to obtain models for our quantified temporal logic. 
Recall that in the semantics of worlds, providing the interpretation of a formula means providing an interpretation of such a formula in \emph{every world}. 

\begin{definition}
Let $\Sigma$ be a many-sorted signature. A \tbf{temporal counterpart $\mW$-model} is defined by $\Tempmodel :=(\mW,\mathrm{T},\sortmodel,\functionmodel)$

\begin{itemize}
   \setlength\itemsep{0.5em}
    \item the triple $(\mW,\sortmodel,\functionmodel)$ is a counterpart $\mW$-model;
    \item every relational presheaf $\relpresheaf{\mW}{X}$ of $\sortmodel$ is equipped with the temporal structure $\mathrm{T}$.
\end{itemize}
\end{definition}
The presence of a temporal structure 
allows us to refine Proposition~\ref{remark caso particolare counterpar model GLV} and Proposition~\ref{proposition ogni counterpar model ci da un counterpart W model} in the context of temporal $\mW$-counterpart models, and to obtain the following correspondence with counterpart models in the sense of \cite{CounterpartSemanticsGLV}.
\begin{theorem}\label{teorema corrispondenza 1-1}
There is a bijective correspondence between counterpart models  $\angbr{W}{\rightsquigarrow,d}$  and temporal $\mW$-counterpart models $(\mW,\mathrm{T},\sortmodel,\functionmodel)$ where $\mW$ is freely generated by the set of objects $\mathbf{ob}(\mW)$ and the class of arrows $\mathrm{T}$ and every relational presheaf of $\sortmodel$ sends a morphism  $\freccia{\omega}{f}{\sigma}$ of  $\mW$ to a relation $\relazione{\interp{\tau}{\Relmodel}_{\sigma}}{\interp{\tau}{\Relmodel}_f}{\interp{\tau}{\Relmodel}_{\omega}}$ whose converse $\relazione{\interp{\tau}{\Relmodel}_{\omega}}{(\interp{\tau}{\Relmodel}_f)^{\dagger}}{\interp{\tau}{\Relmodel}_{\sigma}}$ is a partial function.
\end{theorem}
\begin{proof}
The construction of a temporal $\mW$-counterpart model from a counterpart model is given exactly as in Proposition \ref{proposition ogni counterpar model ci da un counterpart W model} with the obvious choice of the temporal structure. What changes is only the definition of a counterpart model from a temporal $\mW$-counterpart model. In fact, while in Proposition \ref{remark caso particolare counterpar model GLV} every arrow $\freccia{\omega_1}{f}{\omega_2}$ of the base category induces an element $(\omega_1,cr_f,\omega_2)\in\rightsquigarrow $, if we start from a temporal $\mW$-counterpart model as in Theorem \ref{teorema corrispondenza 1-1}, we define an elements $(\omega_1,cr_f,\omega_2)\in\rightsquigarrow $ only for those arrows $f\in \mathrm{T}$ of the temporal structure. It is direct to check that these two constructions provide a bijective correspondence.
\end{proof}

Now we show how terms-in-context and formulae-in-context are interpreted in our model. We start by noticing that, by definition of temporal $\mW$-counterpart model, the interpretation of sorts, function and relation symbols is already fixed.
Therefore, for the rest of this section we fix a temporal $\mW$-counterpart model  $\Tempmodel =(\mW,\mathrm{T},\sortmodel,\functionmodel)$. 

First of all, given a first-order context $\Gamma= \context{x}{\tau}{n}$,  we denote by 
$$ \interp{\Gamma}{\Relmodel}:=\interp{\tau_1}{\Relmodel}\times \dots\times \interp{\tau_n}{\Relmodel}$$
the relational presheaf associated to the context $\Gamma$ via the counterpart $\mW$-model $(\mW,\sortmodel,\functionmodel)$. Then, given a second-order context $\Delta=\context{\chi}{\tau}{m}$, we denote by
$$ \interp{\Delta}{\Relmodel}:=\powerelpresheaf{\interp{\tau_1}{\Relmodel}}\times \dots \times\powerelpresheaf{\interp{\tau_m}{\Relmodel}}$$
where $\powerelpresheaf{\interp{\tau_i}{\Relmodel}}$ denotes the relational power-set presheaf of $\interp{\tau_i}{\Relmodel}$, see Definition \ref{definizione relational powerset}. Therefore we define the interpretation
$$\interp{\Gamma,\Delta}{\Relmodel}=\interp{\Gamma}{\Relmodel}\times \interp{\Delta}{\Relmodel}.$$

Recall that a term in a given context $[\Gamma,\Delta]\; t:\tau $ is then interpreted as a morphism of relational presheaves, defined as
\begin{itemize}
   \item if $t=x_i$ then $\interp{t}{\Relmodel}$ is the projection $\freccia{\interp{\Gamma,\Delta}{\Relmodel}}{\pi_i}{\interp{\tau}{\Relmodel}}$;
    \item if $t=f(t_1,\dots,t_k)$, then $\interp{t}{\Relmodel}$ is given by the composition 
    \[
    \xymatrix@+2pc{
  \interp{\Gamma,\Delta}{\Relmodel} \ar[rr]^{\angbr{\interp{t_1}{\Relmodel},\dots}{\interp{t_k}{\Relmodel}}}& &\interp{\Gamma' }{\Relmodel} \ar[r]^{\mI(f)}& \interp{\tau}{\Relmodel}.
    } \]
\end{itemize}
The interpretation of a given formula in context $[\Gamma] \; \phi $ has to be defined for each world, in line with the usual Kripke-style semantics. Therefore, the interpretation of $[\Gamma] \; \phi $ is  defined as a classical attribute
\[\interp{[\Gamma] \; \phi}{\Tempmodel}=\{\interp{[\Gamma] \; \phi}{\Tempmodel}_{\omega}\}_{\omega\in \mW}\]
of the relational presheaf $\interp{\Gamma}{\Tempmodel}$, where every $\interp{[\Gamma] \; \phi}{\Tempmodel}_{\omega}$ is a subset of $\interp{\Gamma}{\Tempmodel}_{\omega}$. Recall that the notation $\interp{\Gamma}{\Tempmodel}_{\omega}$ indicates the set given by the evaluation $\interp{\Gamma}{\Tempmodel}$ at $\omega$.

Moreover, we start defining the interpretation of a formula at a given fixed world, and we use induction on the structure on $\phi$ as usual. 

The interpretation of standard formulae at a given world $\omega$  is given as follows
\begin{itemize}
   \setlength\itemsep{0.5em}

\item $\interp{[\Gamma,\Delta] \;\top }{\Relmodel}_{\omega}:=\interp{\Gamma,\Delta}{\Relmodel}_{\omega}$;
\item $\interp{[\Gamma,\Delta] \;\bot }{\Relmodel}_{\omega}:=\emptyset$;

\item $\interp{[\Gamma,\Delta] \;\varepsilon\in_{\tau} \chi }{\Relmodel}_{\omega}:= \angbr{\pi_{\varepsilon}}{\pi_{\chi}}_{\omega}^* (\in_{\interp{\tau}{\Relmodel}}(\omega))$ where $\freccia{\interp{[\Gamma,\Delta]}{\Relmodel}}{\angbr{\pi_{\varepsilon}}{\pi_{\chi}}}{\interp{[\varepsilon:\tau, \chi:\tau}{\Relmodel}}$ are the opportune projections;
\item for every $\phi$, then $\interp{\Gamma,\Delta] \; \neg \phi}{\Relmodel}_{\omega}:= \overline{\interp{\Gamma,\Delta] \; \phi}{\Relmodel}_{\omega} }$ where $\overline{(-)}$ denotes the set-theoretical complements, i.e. $\interp{\Gamma,\Delta] \; \neg \phi}{\Relmodel}_{\omega} := \interp{\Gamma,\Delta}{\Relmodel}_{\omega} \setminus \interp{\Gamma,\Delta] \; \phi}{\Relmodel}_{\omega}$.
    \end{itemize}
We also have that   
    \begin{itemize}
       \setlength\itemsep{0.5em}

       \item $\interp{[\Gamma,\Delta]\; \psi \vee \psi}{\Relmodel}_{\omega}:=\interp{[\Gamma,\Delta]\; \psi}{\Tempmodel}_{\omega}\cup \interp{[\Gamma,\Delta]\;  \psi }{\Relmodel}_{\omega}$;
     \item $\interp{[\Gamma,\Delta]\; \psi \wedge \psi}{\Relmodel}_{\omega}:=\interp{[\Gamma,\Delta]\; \psi}{\Tempmodel}_{\omega}\cap \interp{[\Gamma,\Delta]\;  \psi }{\Relmodel}_{\omega}$;
     \item $\interp{[\Gamma,\Delta]\; \forall_{\tau} y. \psi}{\Relmodel}_{\omega}:=\{ a\in \interp{\Gamma,\Delta}{\Relmodel}(\omega)\;| \forall b\in \interp{\tau}{\Relmodel}_{\omega} \mbox{ we have } (a,b)\in \interp{[\Gamma,y:\tau,\Delta]\;\psi}{\Tempmodel}_{\omega}\}$;
 \item $\interp{[\Gamma,\Delta]\; \forall_{\tau} \chi. \psi}{\Relmodel}_{\omega}:=\{ a\in \interp{\Gamma,\Delta}{\Relmodel}(\omega)\;| \forall b\in \powerelpresheaf{\interp{\tau}{\Relmodel}}(\omega) \mbox{ we have } (a,b)\in \interp{[\Gamma,\Delta,\chi: \tau]\;\psi}{\Tempmodel}_{\omega}\}$;
    \item $\interp{[\Gamma,\Delta]\; \exists_{\tau} y. \psi}{\Relmodel}_{\omega}:=\{ a\in \interp{\Gamma,\Delta}{\Relmodel}(\omega)\;| \exists b\in \interp{\tau}{\Relmodel}_{\omega} \mbox{ such that } (a,b)\in \interp{[\Gamma,y:\tau,\Delta]\;\psi}{\Tempmodel}_{\omega}\}$;
 \item $\interp{[\Gamma,\Delta]\; \exists_{\tau} \chi. \psi}{\Relmodel}_{\omega}:=\{ a\in \interp{\Gamma,\Delta}{\Relmodel}(\omega)\;| \exists b\in \powerelpresheaf{\interp{\tau}{\Relmodel}}(\omega) \mbox{ such that } (a,b)\in \interp{[\Gamma,\Delta,\chi: \tau]\;\psi}{\Tempmodel}_{\omega}\}$.
\end{itemize}
Finally, we have the interpretation of formulae in which temporal operators occur
\begin{itemize}
   \setlength\itemsep{0.5em}

    \item $\interp{[\Gamma,\Delta] \;\nextoperator \psi }{\Relmodel}_{\omega}:= \nextoperator\interp{[\Gamma,\Delta] \; \psi }{\Relmodel}_{\omega}$;
      \item $\interp{[\Gamma,\Delta] \;\psi \untiloperator \psi }{\Relmodel}_{\omega}:= \interp{[\Gamma,\Delta] \; \psi }{\Relmodel}_{\omega}\untiloperator\interp{[\Gamma,\Delta] \; \psi }{\Relmodel}_{\omega} $;
      \item $\interp{[\Gamma,\Delta] \;\psi \wuntiloperator \psi }{\Relmodel}_{\omega}:= \interp{[\Gamma,\Delta] \; \psi }{\Relmodel}_{\omega}\wuntiloperator\interp{[\Gamma,\Delta] \; \psi }{\Relmodel}_{\omega} $;
    \item $\interp{[\Gamma,\Delta] \;\sq \psi }{\Relmodel}_{\omega}:= \sq\interp{[\Gamma,\Delta] \; \psi }{\Relmodel}_{\omega}$;
        \item $\interp{[\Gamma,\Delta] \;\dm \psi }{\Relmodel}_{\omega}:= \dm\interp{[\Gamma,\Delta] \; \psi }{\Relmodel}_{\omega}$;

\end{itemize}
where $\nextoperator,\untiloperator,\wuntiloperator,\dm,\sq $ are the temporal operators induced by the temporal structure $\mathrm{T}$.

We conclude by comparing our semantics with that introduced for counterpart models, and we refer to \cite{CounterpartSemanticsGLV} for all the details about that. In particular we can employ the correspondence of Theorem \ref{teorema corrispondenza 1-1} to conclude that the semantics introduced in \cite{CounterpartSemanticsGLV} and the one we present for temporal counterpart models are equivalent.
\begin{theorem}
A formula is satisfied by a counterpart model if and only if it is satisfied by the corresponding $\mW$-temporal counterpart model.
\end{theorem}

\begin{example}
The notion of temporal $\mW$-counterpart model allows us to obtain as particular instance the semantics for standard LTL by considering as temporal structure the free category generated by a linear flow of time as in Example \ref{example LTL}. In particular, the interpretation of the operators $\nextoperator,\dm,\sq, \untiloperator,\wuntiloperator$ coincide with the usual one of LTL.
\end{example}

\begin{example}
Let us consider our running example of Figure \ref{Figura grafi e morfismi parziali} and the temporal structure given by the morphism $\mathrm{T}:=\{f_0,f_1,f_2,f_3\}$. Now let us consider the property  $[y:\tau_N] \;\exists_{\tau_N}x. (((x\neq y) \wedge\nextoperator (x=y))$ presented in Example \ref{example prop. sintat.}: we have that
\begin{itemize}
    \item $\interp{[y:\tau_N] \;\exists_{\tau_N}x. ((x\neq y) \wedge\nextoperator (x=y)) }{\Relmodel}_{\omega_0}=\{n_0,n_2\}$ is the set of nodes of the graph $\mathbf{G}_0$ that will be identified at the next step, i.e. at the world $\omega_1$;
    \item $\interp{[y:\tau_N]\;\exists_{\tau_N}x. ((x\neq y) \wedge\nextoperator (x=y)) }{\Relmodel}_{\omega_1}=\{n_3,n_4\}$ is the set of nodes of the graph $\mathbf{G}_1$ that will be identified at the next step, i.e. at the world $\omega_2$;
    \item $\interp{[y:\tau_N] \;\exists_{\tau_N}x. ((x\neq y) \wedge\nextoperator (x=y)) }{\Relmodel}_{\omega_2}=\emptyset$ is the set of nodes of the graph $\mathbf{G}_2$ that will be identified at the next step.
\end{itemize}

\end{example}
\begin{example}
Recall the toy example presented in the introduction, i.e, the model with two states $s_0$ and $s_1$. In order to describe it, consider a one-sorted signature $\Sigma=\{\tau\}$ with no function symbols, and the free category $\mathcal{S}$ generated by the following diagram 
\[
\xymatrix{
s_0  \ar@<.5ex>[r]^{f_0} & s_1\ar@<.5ex>[l]^{f_1}
}
\]
and the relational presheaf $\relpresheaf{\mathcal{S}}{D}$ 
such that $D_{s_0}=\{i \}$, $D_{s_1}=\emptyset$, and both $D_{f_0}$ and $D_{f_1}$ are the empty relation. Then consider the counterpart model $(\mathcal{S},\mathrm{T}:=\{f_0,f_1\},D,\emptyset)$.
One of the main advantages of the counterpart semantics is the possibility to deal with processes \emph{destroying elements}. In this setting an interesting formula is $[x:\tau] \; \mathbf{present}(x)\wedge\nextoperator(\nextoperator \mathbf{present}(x))$, i.e. meaning that there exists an entity at a given world that has a counterpart after two steps. If we consider its interpretation at world $s_0$ in the model $(\mathcal{S},\{f_0,f_1\},D,\emptyset)$, it is direct to check that 
$$\interp{[x:\tau] \; \mathbf{present}(x)\wedge\nextoperator(\nextoperator \mathbf{present}(x))}{\Tempmodel}_{s_0}=\emptyset.$$
This is exactly what expected, since it essentially means that entity $i$  has no counterpart at the world $s_0$ after two steps, even if it clearly belong to the world relative to $s_0$.
\end{example}
\begin{example}[Deallocation]\label{esempio deallocation}
The creation and destruction of entities has attracted the interest of various authors
(e.g. \cite{distefano2002,Yahav2006}) as a means for reasoning about the allocation and deallocation of resources or processes.
Our logic does not offer an explicit mechanism for this purpose. Nevertheless, as we have shown in
Example \ref{example prop. sintat.}, we can easily derive a predicate regarding the presence of an entity in a certain world
as $ \mathbf{present}_{\tau} (x)$. Using this predicate together with the next-time modality, we can
reason about the preservation and deallocation of some entities after one step of evolution of the system as
$\mathbf{nextStepPreserved}(x)\equiv\mathbf{present}_{\tau} (x)\wedge \nextoperator ( \mathbf{present}_{\tau} (x))$ 
and
$\mathbf{nextStepDeallocated}(x)\equiv\mathbf{present}_{\tau} (x)\wedge \nextoperator (\neg \mathbf{present}_{\tau} (x))$.

Now we provide an interpretation 
of these two formulae for 
our running example 
in Figure~\ref{Figura grafi e morfismi parziali} and the temporal structure given by the morphism $\mathrm{T}:=\{f_0,f_1,f_2,f_3\}$. Then 
\begin{itemize}
     \item $\interp{[x:\tau_E]\;\mathbf{nextstepPreserved}(x) }{\Relmodel}_{\omega_0}=\{e_0,e_1\}$ is the set of edges of $\mathbf{G}_0$ that survive at the next steps;
     \item $\interp{[x:\tau_E]\; \mathbf{nextstepPreserved}(x) }{\Relmodel}_{\omega_1}=\emptyset$ is the set of edges of $\mathbf{G}_1$ that survive at the next steps;
     \item $\interp{[x:\tau_E]\;\mathbf{nextstepPreserved}(x) }{\Relmodel}_{\omega_0}=\{e_5\}$ is the set of edges of $\mathbf{G}_2$ that survive at the next steps.
\end{itemize}
Notice that $\interp{[x:\tau_E]\; \mathbf{nextstepPreserved}(x) }{\Relmodel}_{\omega_1}=\emptyset$ because we have that $d(f_1)$ forgets the arrow $e_4$, while $d(f_2)$ forgets the arrow $e_3$. This follows from  our definition of the next step operator, where we require that a given property has to hold for \emph{every step of length one}. Then, we conclude considering the case of next-step deallocation
\begin{itemize}
    \item $\interp{[x:\tau_E]\;\mathbf{nextstepDeallocated}(x) }{\Relmodel}_{\omega_0}=\{e_2\}$ is the set of edges of $\mathbf{G}_0$ that are deallocated at the next steps;
     \item $\interp{[x:\tau_E]\; \mathbf{nextstepDeallocated}(x)) }{\Relmodel}_{\omega_1}=\emptyset$ is the set of edges of $\mathbf{G}_1$ that are deallocated at the next steps;
     \item $\interp{[x:\tau_E]\;\mathbf{nextstepDeallocated}(x) }{\Relmodel}_{\omega_0}=\emptyset$ is the set of edges of $\mathbf{G}_2$ that are deallocated at the next steps.
\end{itemize}
\end{example}

\section{Conclusions and future works}




In the paper we presented a counterpart semantics for (linear time) temporal logics that is based on relational presheaves.
Starting points were previous works on modal logics, namely the set-theoretical counterpart semantics in~\cite{CounterpartSemanticsGLV} and the functional presheaves model for Kripke frames in~\cite{GhilardiMeloni1988}, and indeed they are both recovered in our framework.

Counterpart semantics offers a suitable solution to the trans-world identity problem, which we have argued to be relevant also from a practitioner point of view. The use of preshaves allows us to easily recover it, as well as to model second-order quantification.

The choice of linear time temporal logics asked for some ingenuity in the way to model the next step operator, with the introduction of what we called temporal structures, as well as in the treatment of negation, which required the use of a syntax for positive formulae. Our presheaf framework may as well recover the semantics of other temporal logics, and in fact we believe that it is general enough that it could be adapted to many different formalisms: indeed, with respect to our focus on partial ones, relational presheaves allows for a very general notion of morphism between worlds, which could e.g. be pivotal for formalisms where non-determinism plays a central role.

From a categorical perspective, our results open two challenging lines for future works: the first one regards the problem of proving a result of completeness. The choice of relational presheaves makes this task tricky to pursue, but it certainly deserves future investigations.
The second regards the study of formal criteria for the semantics of quantified temporal logic in the spirit of categorical logic, where models are thought of as opportune morphisms. A possible solution could be presenting temporal models as morphisms of opportune Lawvere doctrines. This is also a non-trivial problem that we are going to deal with in future work.

\bibliographystyle{abbrv}
\bibliography{biblio_davide}
\end{document}

%% file: macro_davide.tex
\newcommand{\tbf}[1]{\textbf{\emph{#1}}}

\def\mA{\mathcal{A}}

\def\mC{\mathcal{C}}

\def\mD{\mathcal{D}}

\def\mF{\mathcal{F}}

\def\mI{\mathcal{I}}
\def\mG{\mathcal{G}}
\def\mW{\mathcal{W}}

\def\Int{\mathbf{Int}}   
\def\Set{\mathbf{Set }}  
\def\Rel{\mathbf{Rel }}  
\def\Drg{\mathbf{Drg}}    
\def\ParDrg{\mathbf{ParDrg}}    
\def\nuovopf{\Psi}    
\def\op{\operatorname{ op}}
\def\ModalHyp{\mathbf{MoHyp}}
\newcommand{\presheafcat}[1]{[#1^{\op},\Set]}
\newcommand{\relpresheafcat}[1]{[#1^{\op},\Rel]}
\newcommand{\Subobjectdoctrine}{\mathrm{Sub}}

\newcommand{\modaldoctrine}[2]{\xymatrix{#2 \colon #1^{\op}  \ar[r] & \Int }}

\newcommand{\presheaf}[2]{\xymatrix{#2 \colon #1^{\op}  \ar[r] & \Set }}
\newcommand{\relpresheaf}[2]{\xymatrix{#2 \colon #1^{\op}  \ar[r] & \Rel }}
\newcommand{\powersetobject}[1]{\mathrm{P} #1}
\newcommand{\membpredicate}[1]{\in_{#1}}
\newcommand{\freccia}[3]{\xymatrix{#2 \colon #1  \ar[r] &  #3}}
\newcommand{\frecciasopra}[3]{\xymatrix{ #1  \ar[r]^{#2} &  #3}}

\newcommand{\duemorfismo}[6]{\xymatrix{
#1^{\op} \ar[rrd]^#2_{}="a" \ar[dd]_{#3^{\op}}\\
&& \Int\\
#5^{\op}  \ar[rru]_#6^{}="b"
\ar_{#4}  "a";"b"}}

\newcommand{\relazione}[3]{\xymatrix@-0.8pc{#2 \colon #1  \ar[rr]\ar@{-|}[r]& &  #3}}

\def\sq{\square}  
\def\dm{\Diamond} 
\def\nextoperator{\mathsf{O}}
\def\untiloperator{\mathsf{U}}
\def\wuntiloperator{\mathsf{W}}

\def\Einv{\Sigma}
\def\Alg{\mathcal{A}}
\def\att{\mathcal{T}}   
\def\relatt{\mathcal{R}} 
\def\equalitypredicate{\Theta}
\def\modallanguage{\mathcal{L}}
\def\catcontext{\mathbf{C}_{\modallanguage}}
\def\modalsystem{\mathbf{S}}
\def\syntdoct{\mathcal{A}_{\modalsystem}}
\newcommand{\interpl}{\ensuremath{ [ \hspace{-0.45ex} |}} 
\newcommand{\interpr}{  \ensuremath{ | \hspace{-0.45ex}]}} 
\newcommand{\interp}[2]{\interpl #1 \interpr^{#2}}
\def\model{\mathfrak{M}}
\def\Relmodel{\Tempmodel}
\def\Tempmodel{\mathfrak{T}}
\newcommand{\sortmodel}{\mathfrak{S}_{\Sigma}}
\newcommand{\functionmodel}{\mathfrak{F}_{\Sigma}}

\def\id{\mathrm{id}}
\newcommand{\angbr}[2]{\langle #1,#2 \rangle}

\newcommand{\MSalgebra}[1]{\mathrm{#1}}

\newcommand{\frecciaparziale}[3]{\xymatrix{#2 \colon #1  \ar@{-^{>}}[r] &  #3}}

\newcommand{\context}[3]{[#1_1:#2_1,\dots,#1_{#3}:#2_{#3}]}
\newcommand{\termincontext}[3]{#1:#2 \; [#3]}
\newcommand{\SynCat}[1]{\mathsf{Con} (#1)}

\newcommand{\ModelCat}[2]{#1\mbox{-}\mathsf{model}(#2)}
\newcommand{\FPfunctors}{\mathsf{FP}}    
\newcommand{\ContModelCat}[2]{\mathsf{count} \mbox{-}#1\mbox{-}\mathsf{model}(#2)}

\newcommand{\powerelpresheaf}[1]{\mathrm{P}(#1)}

\newcommand{\cammini}[2]{\mathsf{path}(#1,#2)}
\newcommand{\domain}[1]{\mathsf{dom}(#1)}
\newcommand{\codomain}[1]{\mathsf{cod}(#1)}

%% file: Modal_doctrines.bbl
\begin{thebibliography}{10}

\bibitem{Awodey04}
S.~Awodey, K.~Kishida, and H.~Kotzsch.
\newblock Topos semantics for higher-order modal logic.
\newblock {\em Logique et Analyse}, 57(228):591--636, 2014.

\bibitem{BaldanCorradini2006}
P.~Baldan, A.~Corradini, B.~K{\"{o}}nig, and A.~Lluch{-}Lafuente.
\newblock A temporal graph logic for verification of graph transformation
  systems.
\newblock In J.~L. Fiadeiro and P.~Schobbens, editors, {\em WADT 2006}, volume
  4409 of {\em LNCS}, pages 1--20. Springer, 2006.

\bibitem{Belardinelli2006QuantifiedML}
F.~Belardinelli.
\newblock {\em Quantified Modal Logic and the Ontology of Physical Objects}.
\newblock PhD thesis, Scuola Normale Superiore of Pisa, 2006.

\bibitem{HandbookModalLogic2007}
P.~Blackburn, J.~F. A.~K. van Benthem, and F.~Wolter, editors.
\newblock {\em Handbook of Modal Logic}, volume~3.
\newblock North-Holland, 2007.

\bibitem{Braner2007FirstorderML}
T.~Bra{\"{u}}ner and S.~Ghilardi.
\newblock First-order modal logic.
\newblock In Blackburn et~al. \cite{HandbookModalLogic2007}, pages 549--620.

\bibitem{Cardelli2002}
L.~Cardelli, P.~Gardner, and G.~Ghelli.
\newblock A spatial logic for querying graphs.
\newblock In P.~Widmayer, S.~Eidenbenz, F.~Triguero, R.~Morales, R.~Conejo, and
  M.~Hennessy, editors, {\em ICALP 2002}, pages 597--610. Springer, 2002.

\bibitem{COURCELLE199012}
B.~Courcelle.
\newblock The monadic second-order logic of graphs. {I}. {R}ecognizable sets of
  finite graphs.
\newblock {\em Information and Computation}, 85(1):12--75, 1990.

\bibitem{COURCELLE97}
B.~Courcelle.
\newblock The expression of graph properties and graph transformations in
  monadic second-order logic.
\newblock In G.~Rozenberg, editor, {\em Handbook of Graph Grammars and
  Computing by Graph Transformations, Volume 1: Foundations}, pages 313--400.
  World Scientific, 1997.

\bibitem{DAWAR2007263}
A.~Dawar, P.~Gardner, and G.~Ghelli.
\newblock Expressiveness and complexity of graph logic.
\newblock {\em Information and Computation}, 205(3):263--310, 2007.

\bibitem{distefano2002}
D.~Distefano, A.~Rensink, and J.~Katoen.
\newblock Model checking birth and death.
\newblock In R.~A. Baeza{-}Yates, U.~Montanari, and N.~Santoro, editors, {\em
  {IFIP} TCS 2002}, volume 223 of {\em {IFIP} Conference Proceedings}, pages
  435--447. Kluwer, 2002.

\bibitem{Franconi2003FixpointEO}
E.~Franconi and D.~Toman.
\newblock Fixpoint extensions of temporal description logics.
\newblock In D.~Calvanese, G.~D. Giacomo, and E.~Franconi, editors, {\em DL
  2003}, volume~81 of {\em {CEUR} Workshop Proceedings}. CEUR-WS.org, 2003.

\bibitem{freyd1990categories}
P.~Freyd and A.~Scedrov.
\newblock {\em Categories, Allegories}.
\newblock Elsevier, 1990.

\bibitem{CounterpartSemanticsGLV}
F.~Gadducci, A.~{Lluch Lafuente}, and A.~Vandin.
\newblock Counterpart semantics for a second-order $\mu$-calculus.
\newblock {\em Fundamenta Informaticae}, 118(1-2):177--205, 2012.

\bibitem{GARDINER199421}
P.~H. Gardiner, C.~E. Martin, and O.~{de Moor}.
\newblock An algebraic construction of predicate transformers.
\newblock {\em Science of Computer Programming}, 22(1):21--44, 1994.

\bibitem{GhilardiMeloni1988}
S.~Ghilardi and G.~Meloni.
\newblock Modal and tense predicate logic: Models in presheaves and categorical
  conceptualization.
\newblock In F.~Borceux, editor, {\em Categorical Algebra and its
  Applications}, volume 1348 of {\em Lecture Notes in Mathematics}, pages
  130--142. Springer, 1988.

\bibitem{GhilardiMeloni1990}
S.~Ghilardi and G.~Meloni.
\newblock Relational and topological semantics for temporal and modal
  predicative logic.
\newblock In G.~Corsi and G.~Sambin, editors, {\em Nuovi problemi della logica
  e della scienza II}, pages 59--77. CLUEB, 1990.

\bibitem{GhilardiMeloni1996}
S.~Ghilardi and G.~Meloni.
\newblock Relational and partial variable sets and basic predicate logic.
\newblock {\em Journal of Symbolic Logic}, 61(3):843--872, 1996.

\bibitem{Hazen1979}
A.~Hazen.
\newblock Counterpart-theoretic semantics for modal logic.
\newblock {\em The Journal of Philosophy}, 76(6):319--338, 1979.

\bibitem{Hodkinson2001}
I.~M. Hodkinson, F.~Wolter, and M.~Zakharyaschev.
\newblock Monodic fragments of first-order temporal logics: 2000-2001 {A.D}.
\newblock In R.~Nieuwenhuis and A.~Voronkov, editors, {\em LPAR 2001}, volume
  2250 of {\em LNCS}, pages 1--23. Springer, 2001.

\bibitem{Jacobs2001}
B.~Jacobs.
\newblock Many-sorted coalgebraic modal logic: A model-theoretic study.
\newblock {\em RAIRO--Theoretical Informatics and Applications}, 35:31--59,
  2001.

\bibitem{jacobs2002}
B.~Jacobs.
\newblock The temporal logic of coalgebras via {G}alois algebras.
\newblock {\em Mathematical Structures in Computer Science}, 12(6), 2002.

\bibitem{Lawvere1969}
F.~Lawvere.
\newblock Adjointness in foundations.
\newblock {\em Dialectica}, 23:281--296, 1969.

\bibitem{Lawvere1969b}
F.~Lawvere.
\newblock Diagonal arguments and cartesian closed categories.
\newblock In P.~J. Hilton, editor, {\em Category Theory, Homology Theory and
  their Applications II}, volume~92 of {\em Lecture Notes in Mathematics}, page
  134–145. Springer, 1969.

\bibitem{CounterpartTheoryLewis}
D.~K. Lewis.
\newblock Counterpart theory and quantified modal logic.
\newblock {\em The Journal of Philosophy}, 65(5):113--126, 1968.

\bibitem{QCFF}
M.~Maietti and G.~Rosolini.
\newblock Quotient completion for the foundation of constructive mathematics.
\newblock {\em Logica Universalis}, 7(3):371--402, 2013.

\bibitem{Meseguer2008}
J.~Meseguer.
\newblock The temporal logic of rewriting: {A} gentle introduction.
\newblock In P.~Degano, R.~D. Nicola, and J.~Meseguer, editors, {\em
  Concurrency, Graphs and Models}, volume 5065 of {\em LNCS}, pages 354--382.
  Springer, 2008.

\bibitem{Niefield2010}
S.~Niefield.
\newblock Lax presheaves and exponentiability.
\newblock {\em Theory and Applications of Categories}, 24(12):288--301, 2010.

\bibitem{Pnueli77}
A.~Pnueli.
\newblock The temporal logic of programs.
\newblock In {\em FOCS 1977}, pages 46--57. {IEEE} Computer Society, 1977.

\bibitem{Reif85}
J.~H. Reif and A.~P. Sistla.
\newblock A multiprocess network logic with temporal and spatial modalities.
\newblock {\em Journal of Computer and System Science}, 30(1):41--53, 1985.

\bibitem{Rensink06}
A.~Rensink.
\newblock Model checking quantified computation tree logic.
\newblock In C.~Baier and H.~Hermanns, editors, {\em CONCUR 2006}, volume 4137
  of {\em LNCS}, pages 110--125. Springer, 2006.

\bibitem{ECRT}
D.~Trotta.
\newblock The existential completion.
\newblock {\em Theory and Applications of Categories}, 35:1576--1607, 2020.

\bibitem{Yahav2006}
E.~Yahav, T.~Reps, M.~Sagiv, and R.~Wilhelm.
\newblock Verifying temporal heap properties specified via evolution logic.
\newblock {\em Logic Journal of the IGPL}, 14(5):755--783, 2006.

\end{thebibliography}
